\documentclass[lettersize,journal]{IEEEtran}
\usepackage{cite}
\usepackage{amsmath,amssymb,amsfonts}
\usepackage{amsmath}
\usepackage{textcomp}
\usepackage{xcolor}
\usepackage{enumerate}   
\usepackage{amsmath}
\usepackage[utf8]{inputenc} 
\usepackage[T1]{fontenc}    
\usepackage{hyperref}       
\usepackage{url}            
\usepackage{booktabs}       
\usepackage{amsfonts}       
\usepackage{nicefrac}       
\usepackage{microtype}      
\usepackage{amsfonts}
\usepackage{amssymb}
\usepackage{graphicx}
\usepackage{fancyhdr,color}
\usepackage{mathtools}
\usepackage{bbm}
\usepackage{bm}
\usepackage{mathabx}
\usepackage{diagbox}
\usepackage{balance}
\usepackage{amsthm}
\usepackage{makecell}

\newcommand*{\medcap}{\mathbin{\scalebox{0.8}{\ensuremath{\bigcap}}}}%
\newcommand*{\medcup}{\mathbin{\scalebox{0.8}{\ensuremath{\bigcup}}}}%
\newtheorem{theorem}{Theorem}
\newtheorem{remark}{Remark}

\newtheorem{proposition}{Proposition}

\usepackage[ruled, lined, linesnumbered, commentsnumbered, longend]{algorithm2e}
\def\BibTeX{{\rm B\kern-.05em{\sc i\kern-.025em b}\kern-.08em
    T\kern-.1667em\lower.7ex\hbox{E}\kern-.125emX}}

\begin{document}

\title{Energy-efficient predictive control for connected, automated driving under localization uncertainty}

\author{Eunhyek Joa \qquad Eric Yongkeun Choi \qquad Francesco Borrelli ~\IEEEmembership{Fellow,~IEEE,}
\thanks{The authors are with the Model Predictive Control Lab,
Department of Mechanical Engineering, University of California at Berkeley (e-mail: e.joa@berkeley.edu; yk90@berkeley.edu; fborrelli@berkeley.edu).}
}

\markboth{Journal of \LaTeX\ Class Files,~Vol.~14, No.~8, August~2021}%
{Shell \MakeLowercase{\textit{et al.}}: A Sample Article Using IEEEtran.cls for IEEE Journals}

\IEEEpubid{0000--0000/00\$00.00~\copyright~2021 IEEE}

\maketitle

\begin{abstract}
This paper presents a data-driven Model Predictive Control (MPC) for energy-efficient urban road driving for connected, automated vehicles. The proposed MPC aims to minimize total energy consumption by controlling the vehicle's longitudinal motion on roads with traffic lights and front vehicles. Its terminal cost function and terminal constraints are learned from data, which consists of the closed-loop state and input trajectories. The terminal cost function represents the remaining energy-to-spend starting from a given terminal state. The terminal constraints are designed to ensure that the controlled vehicle timely crosses the upcoming traffic light, adheres to traffic laws, and accounts for the front vehicles. We validate the effectiveness of our method through both simulations and \textcolor{black}{vehicle-in-the-loop experiments}, demonstrating $\textbf{19\%}$ improvement in average energy \textcolor{black}{efficiency} compared to conventional approaches that involve solving a long-horizon optimal control problem for speed planning and employing a separate controller for speed tracking.
\end{abstract}


\section{Introduction}
\IEEEPARstart{E}{xtensive} research has been conducted to study how Connected and Automated Vehicles (CAVs) can improve our daily life driving experiences, covering areas from fleet dispatching \cite{qiu2023dispatching} and routing \cite{mostafizi2022routing, guanetti2019routing} to urban driving \cite{guo2019urban} and parking \cite{wang2021parking} scenarios.
The benefits of CAV technology include improved traffic safety, better road utilization, and reduced energy consumption and have been well demonstrated across various driving scenarios \cite{ELLIOTT2019Recent, guanetti2018CAVs, karbasi2022Investigating}.
In this paper, we focus on the potential of CAV technology to enhance energy efficiency for urban road driving.


Existing literature primarily explores optimization-based algorithms to minimize energy consumption while also prioritizing safety, compliance with traffic laws, and driving comfort.
They have suggested a hierarchical control architecture to manage this complexity, dividing responsibilities between an upper layer that plans the vehicle's future trajectory and a lower layer that executes the planned trajectory through direct vehicle control.
The upper layer has been designed using various optimization control methods.
In \cite{deshpande2022DP, bae2019real}, a Dynamic Programming (DP) method is used for the upper layer planning algorithms, with Model Predictive Control (MPC) being employed at the lower layer.
Pontryagin's Minimum Principle (PMP) method has been used in \cite{ard2023VILCAV, han2023energy}, and convex optimization is used in \cite{bo2023MPC, chada2020ecological} for their planning algorithms.
These studies often assume ideal conditions, such as empty roads at the planning stage, forcing the lower layer to deviate from the vehicle's trajectory in real-time based on immediate conditions, such as the behavior of the vehicle ahead.
However, deviations from the planned trajectory can lead to unforeseen consequences.
For instance, the ego vehicle may fail to pass a traffic light as planned due to unforeseen traffic ahead that was not accounted for at the planning stage.
Moreover, the inherent drawbacks of having two separate control layers include potential delays and latency, which can lead to spatial and temporal misalignments and, consequently, discrepancies in the control layers' decisions.
Therefore, we propose a unified control architecture with a single MPC that directly computes longitudinal acceleration while ensuring safety.



This paper aims to address key practical implementation issues associated with CAVs. Achieving high-precision localization typically requires expensive technologies like differential GPS or LiDAR. Recognizing the cost constraints linked to the imminent deployment of CAVs, our proposed controller considers localization uncertainty. Existing research often assumes exact vehicle positioning, an assumption that may not be viable with the forthcoming wave of CAV technologies that utilize less accurate localization methods. To address these issues, the proposed MPC is a robust MPC with respect to localization uncertainties. Additionally, the proposed MPC minimizes the expected energy consumption to maintain energy efficiency under these conditions.
\IEEEpubidadjcol

In our conference paper \cite{joa2024eco}, we proposed a data-driven MPC that solves the given problem under free-flow conditions and experimentally validates the algorithm under one scenario. In this paper, we extend the algorithm to handle urban road scenarios with a front vehicle and investigate the proposed algorithm with multiple scenarios in simulations and vehicle experiments.
Our contributions are summarized as follows:
\begin{itemize}
    \item  We propose a novel data-driven MPC to solve an energy-efficient urban road driving problem for connected, automated vehicles considering localization uncertainty.
    \item We propose a unified MPC controller to replace the conventional hierarchical approach, effectively mitigating compound errors and latency issues.
    \item We experimentally demonstrate the energy saving of the proposed algorithm through \textcolor{black}{vehicle-in-the-loop tests} where the actual test vehicle is controlled to complete a given route within a user-defined duration under localization uncertainty while interacting with a virtual front vehicle and deterministic traffic lights behavior.
\end{itemize}

\textit{Notation:} 
Throughout the paper, we use the following notation.
$\mathbf{0}^{n \times m}$ represents an $n$-by-$m$ zero matrix.
The positive semi-definite matrix \(P\) is denoted as \(P\succeq0\).
The Minkowski sum of two sets is denoted as \(\mathcal{X} \oplus \mathcal{Y} = \{x+y: x \in \mathcal{X}, y \in \mathcal{Y}\}\).
The Pontryagin difference between two sets is defined as \(\mathcal{X} \ominus \mathcal{Y} = \{x: x+y \in \mathcal{X}, \forall y \in \mathcal{Y}\}\).
The m-th column vector of a matrix \(H\) is denoted as \([H]_m\). The m-th component of a vector \(h\) is \([h]_m\). 
The notation \(x_{l:m}\) means the sequence of the variable \(x\) from time step \(l\) to time step \(m\).

\section{Problem Setup} \label{sec: problem setup}
In this section, we introduce the problem setup of energy-efficient urban driving for connected, automated vehicles under localization uncertainty.
We make the following assumptions.
\begin{itemize}
    \item \textbf{Route Availability:} We assume the existence of a predefined route from a starting point to a goal point, which is typical for automated vehicles.
    
    \item \textbf{Deterministic Traffic Lights:} We assume that all traffic light cycles are deterministic. This assumption is valid when adaptive traffic control systems are either not equipped or deactivated. If traffic light cycles vary, their stochastic nature can be considered, as discussed in \cite{sun2020optimal}.
    
    \item \textbf{V2I Communication Infrastructure:} The infrastructure provides essential data such as distance and traffic flow speed between consecutive traffic lights through V2I communication. This information enables the calculation of the time required to traverse a traffic light from the previous traffic light. Such data aligns with the V2X communication standard in the US, as detailed in \cite{spatstandard}. \textcolor{black}{Additionally, we assume that V2I communication error is negligible, i.e., no packet loss and latency. }
    
    \item \textbf{Negligible Road Grade Impact:} We assume that the road grade along the route is sufficiently small to disregard gravitational potential energy. If the road grade is known, the associated gravitational potential energy could be factored into the terminal cost as discussed in Sec. \ref{sec: learning V function}.
\end{itemize}



\subsection{Vehicle Model, Measurement Model, and Observer}
We model vehicle longitudinal dynamics as a double integrator:
\begin{equation} \label{eq: longitudinal dynamics}
\begin{aligned}
    & \mathbf{x}(t) = \begin{bmatrix} s(t) & v_x(t)\end{bmatrix}^\top, ~ u(t) = a_x(t), \\
    & \Dot{\mathbf{x}}(t) = \begin{bmatrix} 0 & 1\\ 0 & 0\end{bmatrix}\mathbf{x}(t) + \begin{bmatrix} 0 \\ 1 \end{bmatrix}u(t),
\end{aligned}
\end{equation}
where $s$ represents a longitudinal position along the centerline of the given route, and $v_x$ and $a_x$ are the ego vehicle's speed and acceleration, respectively.
Throughout the paper, we will refer to $s$ as the position for brevity. 
Forces due to road grade, air drag, and rolling resistance are not included in \eqref{eq: longitudinal dynamics} because $a_x$ is net longitudinal acceleration. 
We regard controlling the vehicle under those forces as the task of the actuator-level controller.
We discretize the dynamics \eqref{eq: longitudinal dynamics} as
\begin{equation} \label{eq:system}
\begin{aligned}
    & \mathbf{x}_{k}=\begin{bmatrix} s_k & v_{x,k} \end{bmatrix}^\top, ~ u_{k} = a_{x,k}, \\
    & \mathbf{x}_{k+1} = \mathbf{A} \mathbf{x}_{k} + \mathbf{B} u_{k},
\end{aligned}
\end{equation}
where $\mathbf{x}_{k}$ denotes the state, and $u_{k}$ denotes the input at time step $k$. $s_k$ represents the position, while $v_x$ and $a_x$ correspond to the vehicle's longitudinal speed and acceleration at time step $k$. The discretization sampling time $T_s$ is 1 sec.

At time step $k$, we measure the states from sensors as:
\begin{equation} \label{eq: sensor measurements}
\begin{aligned}
    & \mathbf{y}_k = \mathbf{C} \mathbf{x}_{k} + \mathbf{D} w_k = \begin{bmatrix} 1 & 0 \\ 0 & 1 \end{bmatrix}\mathbf{x}_{k} + \begin{bmatrix} 1 \\ 0 \end{bmatrix}w_k,
\end{aligned}
\end{equation}
where $\mathbf{y}_k$ denotes a measurement at time step $k$, which consists of the position $s_k$ and the speed $v_{x,k}$. The position $s_k$ is measured by a localization module, and $w_k$ represents its localization uncertainty. The speed $v_{x,k}$ is measured by vehicle wheel encoders.
We assume the measurement noises of the vehicle speed are small enough to be neglected.

We assume that the longitudinal localization uncertainty $w$ is a random variable with distribution $p_w(w)$ and contained in a segment $\mathcal{W}$ as follows:
\begin{equation} \label{eq: gps error}
\begin{aligned}
    & w_k \sim p_w(w), ~ w_k \in \mathcal{W}.
\end{aligned}
\end{equation}
Let $w_\text{min}$ and $w_\text{max}$ denote the extremes of the segment $\mathcal{W}$. As the $\mathcal{W} \subset \mathbb{R}$, $\mathcal{W}$ can be written as $\{w ~|~ w_\text{min} \leq w \leq w_\text{max} \}$.
In practice, the localization uncertainty may be designed as a Gaussian distribution, which is not bounded. In this case, a high confidence interval can be used to approximate $\mathcal{W}$.

We design a discrete state observer \cite{ogata2010modern} with a predictor to estimate the state $\mathbf{x}_{k}$ as follows:
\begin{equation} \label{eq: observer equation}
\begin{aligned}
    & \hat{\mathbf{x}}_{0} = \mathbf{y}_0, \\
    & \hat{\mathbf{x}}_{k+1} = \mathbf{A} \hat{\mathbf{x}}_{k} + \mathbf{B} u_{k} + \begin{bmatrix} L & 0 \\ 0 & 1 \end{bmatrix}\biggr(\mathbf{y}_{k+1} - \mathbf{C} ( \mathbf{A} \hat{\mathbf{x}}_{k} + \mathbf{B} u_{k})\biggr), \\
    & ~~~~~~ =  \mathbf{A} \hat{\mathbf{x}}_{k} + \mathbf{B} u_{k} + \mathbf{D}n_k,
\end{aligned}
\end{equation}
where $\hat{\mathbf{x}}_{k}=\begin{bmatrix} \hat{s}_{k} & \hat{v}_{x,k} \end{bmatrix}^\top$ is an estimated state, $L$ is an observer gain, and $n_k = L(s_{k} -\hat{s}_{k}) + Lw_{k+1}$, which is a lumped noise. Note that as the speed measurement is accurate, the observer \eqref{eq: observer equation} is designed to set the current speed estimate to the current speed measurement, i.e., $\hat{v}_{x,k} = \begin{bmatrix} 0 & 1\end{bmatrix}y_k$. On the other hand, as the position measurement is not accurate, the observer \eqref{eq: observer equation} is designed to suppress the localization uncertainty.

\begin{proposition} \label{prop: delta s and lumped noise}
    \cite[Prop. 1]{joa2024eco} Let $\Delta s_k = s_{k} -\hat{s}_{k}$. Then, $\Delta s_k \in \mathcal{W}$ and $n_k \in 2L\mathcal{W}$ for all realizations of noise that satisfies \eqref{eq: gps error}, i.e., $\forall w_k \sim p_w(w), ~ w_k \in \mathcal{W}$.
\end{proposition}
\begin{proof}
See \cite[Prop. 1]{joa2024eco}
\end{proof}
From Proposition \ref{prop: delta s and lumped noise}, the lumped noise $n_k$ in \eqref{eq: observer equation} is a random variable with bounded support $2L\mathcal{W}$. 
Let $p_n(n)$ denote the corresponding probability density function.
Then, the lumped noise $n_k$ is a random variable with the probability density function $p_n(n)$ and bounded support $2L\mathcal{W}$:
\begin{equation} \label{eq: lumped noise}
\begin{aligned}
    & n_k \sim p_n(n), ~ n_k \in 2L\mathcal{W}.
\end{aligned}
\end{equation}

\subsection{Energy Model}
We define the energy $E$ as the sum of the energy stored in a battery/fuel tank and the kinetic energy $E_\text{kin}$ as:
\begin{equation} \label{eq: energy definition}
\begin{aligned}
    & E(t) = E_\text{stor}(t) + E_\text{kin}(t).
\end{aligned}
\end{equation}
The energy consumption at time step $k$ is defined as the change in energy between time step $k$ and $k+1$ as follows:
\begin{equation} \label{eq: power definition}
\begin{aligned}
    & \Delta E_k = E(kT_s) - E((k+1)T_s).
\end{aligned}
\end{equation}

We model $\Delta E_k$ \eqref{eq: power definition} as a parameterized function $\ell(\mathbf{x}_k, u_k)$ which is defined as follows:
\begin{equation}
    \begin{aligned}
        & \ell(\mathbf{x}_k, u_k) = \begin{bmatrix} v_{x,k} & u_k & 1 \end{bmatrix} \mathbf{P} \begin{bmatrix} v_{x,k} \\ u_k \\ 1 \end{bmatrix}, \\
    \end{aligned}
\end{equation}
where $\mathbf{P}$ is a symmetric matrix in $\mathbb{R}^{3\times3}$, which includes parameters. 
Additionally, $\mathbf{P}$ is positive semidefinite, ensuring that the parametrized function remains nonnegative, in accordance with the nonnegative nature of energy consumption in \eqref{eq: power definition}.
To obtain the parameters in $\mathbf{P}$, we collect the energy consumption data from our test vehicle during urban road driving. After collecting data, we solve the following regression problem to obtain the parameters in $\mathbf{P}$.:
\begin{equation} \label{eq: energy cost regression}
\begin{aligned}
   & \min_{\mathbf{P} \in \mathbb{R}^{3\times3}} \sum_{k=0}^{T_\text{data}} \lVert \Delta E_k - \ell(\mathbf{x}_k, u_k) \rVert_2^2 \\
   & ~~~\text{s.t.,} ~~\mathbf{P} \succeq 0,
\end{aligned}
\end{equation} 
where $T_\text{data}$ is the end time of the data.
This is a semidefinite programming that can be solved with convex optimization solvers. 
It is noteworthy that our parameterized cost does not depend on the position whose actual value cannot be obtained.
Thus, the following holds:
\begin{equation} \label{eq: CE stage cost}
\begin{aligned}
    & \ell(\mathbf{x}_k, u_k) = \ell(\hat{\mathbf{x}}_k, u_k).
\end{aligned}
\end{equation}

\subsection{Prediction Module of the Front Vehicle}
Let $s^{pv}_k$ and $v^{pv}_k$ denote the longitudinal position and speed of the front vehicle at time step $k$, respectively. 
$d_k$ denotes the distance between the front vehicle and the ego vehicle, which can be derived as $s^{pv}_k - s_k$.
At time step $k$, $d_k$ and $v^{pv}_k$ are measured by proximity sensors such as radars.
We assume that the measurement noises of these sensors are negligible.

In this paper, we assume that there is a given prediction module that predicts the future behavior of the front vehicle.
Let $\hat{s}^{pv}_{i|k}$ and $\hat{v}^{pv}_{i|k}$ respectively denote the front vehicle's position and speed at time step $k+i$ predicted at time step $k$.
The prediction module computes a $T_p$-step predicted state sequence for $s^{pv}$ and $v^{pv}$ respectively denoted by $\{\hat{s}^{pv}_{i|k}\}_{i=0}^{T_p}$ and $\{\hat{v}^{pv}_{i|k}\}_{i=0}^{T_p}$, with initial conditions $\hat{s}^{pv}_{0|k}=\hat{s}_k + d_k$ and $\hat{v}^{pv}_{0|k} = v^{pv}_k$.


In Sec. \ref{sec: learning constraints}, we need the predicted behavior of the front vehicle beyond the maximum prediction horizon $T_p$. In this case, we extrapolate the prediction using a constant speed assumption as follows:
\begin{equation}
\begin{aligned}
    & \hat{s}^{pv}_{i+1|k} = \hat{s}^{pv}_{i|k} + \hat{v}^{pv}_{T_p|k}T_s,\\
    & \hat{v}^{pv}_{i|k} = \hat{v}^{pv}_{T_p|k}, \\
    & \forall i \geq T_p.
\end{aligned}
\end{equation}

In the literature, various methodologies have been employed to analyze the behavior prediction of the front vehicle. These approaches include scenario-based \cite{jeong2019target}, physics-based \cite{lefkopoulos2020interaction}, pattern-based \cite{nayakanti2023wayformer}, and planning-based \cite{wang2021game}. Interested readers can refer to \cite{karle2022scenario}.

\subsection{Collision Avoidance Constraints}
To avoid a collision with the front vehicle, the system \eqref{eq:system} is subject to the following constraints:
\begin{equation} \label{eq: collision avoidance constraint}
    \begin{aligned}
        & \mathbf{x}_k  \in \mathcal{C}(s^{pv}_k,v^{pv}_k) \\
        & \qquad = \{\mathbf{x}_k | s^{pv}_k + v^{pv}_k \cdot \mathrm{TTC} \geq \begin{bmatrix} 1 & \mathrm{TTC} \end{bmatrix}\mathbf{x}_k + d_0 \},
    \end{aligned}
\end{equation}
where $d_0$ denotes the minimum distance considering the size of vehicles, and $\mathrm{TTC}$ denotes time-to-collision (TTC) \cite{moon2009design}.
This constraint \eqref{eq: collision avoidance constraint} imposes that the distance between the ego vehicle and the front vehicle maintains at least a $1$ second TTC gap.
Throughout the paper, $d_0 = 5 \text{m}$ and $\mathrm{TTC}= 1 \text{s}$.


\subsection{State and Input Constraints}
The system (\ref{eq:system}) is subject to the following constraints: 
\begin{equation} \label{eq:constraints}
\begin{aligned}
    & \mathbf{x}_{k} \in \mathcal{X} = \{\mathbf{x}_{k} ~|~ 0 \leq \begin{bmatrix} 0 & 1 \end{bmatrix} \mathbf{x}_{k} \leq v_{x, \max} \}, \\
    & u_k \in \mathcal{U} = \{u ~|~ a_{x, \min} \leq u \leq a_{x, \max} \}, \\
    & \forall k \geq 0, ~\forall w_k \in \mathcal{W},
\end{aligned} 
\end{equation}
where $v_{x, \max}$ is the maximum allowable speed, and $a_{x, \min / \max}$ is the minimum/maximum acceleration.
\(\mathcal{X}\) is a convex polyhedron, and \(\mathcal{U}\) is a polytope.
Note that as the state constraints are imposed on the vehicle speed, $\mathbf{x}_{k} \in \mathcal{X}$ is identical to $\hat{\mathbf{x}}_{k} \in \mathcal{X}$.

\subsection{Traffic Light Crossing Time Constraints} \label{sec: tlctc}
We impose traffic light crossing time constraints on the system \eqref{eq:system} to ensure that the vehicle crosses each traffic light within a predetermined time step. This time step is chosen to avoid disrupting traffic flow by the controlled vehicle.
To find such a time step, we utilize a high-level green wave search module from \cite{highlevelgreenwave1} with one modification: the cost function is adjusted to be quadratic, penalizing a deviation between the traffic flow speed, which is available through V2I, and the average planned speed for each road segment. 
\begin{remark}
    According to \cite[Table 1]{sciarretta2015optimal}, reducing vehicle speed is correlated with lower energy consumption. Consequently, in the absence of constraints mandating adherence to traffic flow speed, the most energy-efficient approach is to maintain a low speed. However, this practice can adversely affect traffic flow, which is socially unacceptable. To mitigate this issue, we propose implementing traffic light crossing time constraints.
\end{remark}

Let $T_l$ denote the predetermined time step required to cross the $l$-th traffic light positioned at $s_{\mathrm{tl},l}$. 
The system (\ref{eq:system}) is subject to the following constraints:
\begin{equation} \label{eq:TL constraint}
    \begin{aligned}
        & \mathbf{x}_{k} \in \mathcal{T}_l = \{\mathbf{x}_{k} ~|~ s_{\mathrm{tl},l} \leq \begin{bmatrix} 1 & 0\end{bmatrix} \mathbf{x}_{k} \}, ~ \forall k \geq T_l, \\
        & l \in \{1,\cdots,N_\mathrm{tl}\},
    \end{aligned}
\end{equation}
where $N_\mathrm{tl}$ denotes the number of traffic lights on the given route.
Note that this constraint \eqref{eq:TL constraint} does not consider traffic light rules; therefore, we need the associated traffic light rule constraints described in the following section.

\subsection{Traffic Light Rule Constraints}
We impose traffic light rule constraints to ensure the vehicle does not run a red light. 
To do that, given state $\mathbf{x}_k$, we find the nearest upcoming traffic light along the given route as follows:
\begin{equation} \label{eq: s_tl(x)}
\begin{aligned}
    & \min_l ~ \lVert \begin{bmatrix} 1 & 0 \end{bmatrix}\mathbf{x}_k  - s_{\mathrm{tl},l}\rVert \\
    & \, ~ \text{s.t.,} ~ \begin{bmatrix} 1 & 0 \end{bmatrix}\mathbf{x}_k \leq s_{\mathrm{tl},l}, \\
    & \,~ ~~~~~~  l \in \{1,\cdots,N_\mathrm{tl}\}.
\end{aligned}    
\end{equation}
After solving the problem \eqref{eq: s_tl(x)}, we obtain the index of the nearest upcoming traffic light $l^\star(\mathbf{x}_k)$.
Let $c_{l,k}$ denote the traffic light phase of the $l$-th traffic light at time step $k$.
Then, if $c_{l^\star(\mathbf{x}_k),k}$ is red, the traffic light rule constraint is imposed and written as follows:
\begin{equation} \label{eq: light rule constraint}
    \begin{aligned}
        & \mathbf{x}_k  \in \mathcal{L}(s_{\mathrm{tl},l^\star(\mathbf{x}_k)}) = \{\mathbf{x}_k ~|~ \begin{bmatrix} 1 & 0 \end{bmatrix} \mathbf{x}_k \leq s_{\mathrm{tl},l^\star(\mathbf{x}_k)}\}.
    \end{aligned}
\end{equation}

\subsection{Energy-efficient Urban Road Driving for Connected, Automated Vehicles under Localization Uncertainty}
The urban driving problem for connected, automated vehicles, aimed at minimizing energy consumption while accounting for localization uncertainty, is a stochastic optimization problem due to stochastic localization uncertainty \eqref{eq: gps error}.
Specifically, it can be formulated as follows:
\begin{equation} \label{eq:ftocp}
\begin{aligned}
    & J^{\star}(\mathbf{x}_S, \{\hat{s}^{pv}_{i|0}\}_{i=0}^{T_{N_\mathrm{tl}}}, \{\hat{v}^{pv}_{i|0}\}_{i=0}^{T_{N_\mathrm{tl}}}) = \\
    & \min_{\Pi(\cdot)} ~\,\mathbb{E}_{w_{0:T_{N_\mathrm{tl}}}}\Biggr[\sum_{k=0}^{T_{N_\mathrm{tl}}} \ell(\hat{\mathbf{x}}_k, ~ \pi_k(\hat{\mathbf{x}}_k)) \Biggr] \\
    & ~ \textnormal{s.t.,} ~~ \mathbf{x}_{k+1}  = \mathbf{A} \mathbf{x}_k  + \mathbf{B} \pi_k (\hat{\mathbf{x}}_k), \\
    & \qquad ~ \mathbf{y}_{k}  = \mathbf{C} \mathbf{x}_k  + \mathbf{D}w_k , ~ w_k  \sim p_w(w), \\
    & \qquad ~ \hat{\mathbf{x}}_{k+1} =\mathbf{A} \hat{\mathbf{x}}_{k} + \mathbf{B} \pi_k (\hat{\mathbf{x}}_k) \\
    & \qquad \qquad ~~~~ + \begin{bmatrix} L & 0 \\ 0 & 1 \end{bmatrix}\biggr(\mathbf{y}_{k+1} - \mathbf{C} ( \mathbf{A} \hat{\mathbf{x}}_{k} + \mathbf{B} \pi_k (\hat{\mathbf{x}}_k))\biggr), \\
    & \qquad ~ \mathbf{x}_0 = \mathbf{x}_S, ~ \hat{\mathbf{x}}_{0} = \mathbf{y}_{0}, \\
    & \qquad ~ \mathbf{x}_{k} \in \mathcal{X}, ~\forall w_k \in \mathcal{W}, \\
    & \qquad ~ \mathbf{x}_{k} \in \mathcal{T}_l, ~\forall k \geq T_l, ~\forall w_k \in \mathcal{W}, \\
    & \qquad ~ \mathbf{x}_k  \in \mathcal{L}(s_{\mathrm{tl},l^\star(\mathbf{x}_k)}), \text{if } c_{l^\star(\mathbf{x}_k),k}=\text{red}, \forall w_k \in \mathcal{W},  \\
    & \qquad ~ \hat{\mathbf{x}}_{k} \in \mathcal{C}(\hat{s}^{pv}_{k|0},\hat{v}^{pv}_{k|0}),  ~ \pi_k(\hat{\mathbf{x}}_{k}) \in \mathcal{U},  ~\forall w_k \in \mathcal{W}, \\
    & \qquad ~ k \in \{0,...,T_{N_\mathrm{tl}}-1\}, ~ l \in \{1,\cdots,N_\mathrm{tl}\}, 
\end{aligned}
\end{equation}
where $T_{N_\mathrm{tl}}$ is the task horizon, which denotes the predetermined time step to cross the last traffic light.
The cost function is an expected sum of the regressed energy consumption stage cost $\ell(\cdot,\cdot)$ in \eqref{eq: energy cost regression} evaluated for the estimated state trajectory. From \eqref{eq: CE stage cost}, this cost is identical to the cost evaluated for the actual state trajectory.
We point out that as the system \eqref{eq: observer equation} is uncertain, the optimal control problem (\ref{eq:ftocp}) consists of finding state feedback policies \(\Pi(\cdot) = \{\pi_0(\cdot),\pi_1(\cdot),...,\pi_{T_f-1}(\cdot)\} \). $\hat{\mathbf{x}}_{k}$ is the argument of the control input policy $\pi_k(\cdot)$. 


\section{Solution approach to problem \eqref{eq:ftocp}} \label{sec: solution approach}
There are three challenges associated with solving \eqref{eq:ftocp}: 
\begin{enumerate}[(C1)]
    \item Optimizing over control policies \(\Pi(\cdot)\) presents an infinite dimensional optimization problem.
    \item The computational demands for solving \eqref{eq:ftocp} become significant when
    $T_{N_\mathrm{tl}}\gg0$.
    \item Minimizing the expected cost in \eqref{eq:ftocp} involves a multivariate integral.
\end{enumerate}

To address (C1), we approximate the control policy by a constant control input, i.e., $\pi_k(\cdot) = u_k$.

To address (C2), we adopt two strategies.
First, we divide the original problem in \eqref{eq:ftocp} into small sub-problems of crossing only the upcoming traffic light within a predetermined time at each time step and solve each sub-problem until the vehicle reaches the end of the predetermined route.
For brevity, let $s_\mathrm{tl}=s_{\mathrm{tl},l^\star(\hat{\mathbf{x}}_k)}$ and $l^\star= l^\star(\hat{\mathbf{x}}_k)$ as we now deal with a single, upcoming traffic light.
Second, we solve a simpler constrained OCP with prediction horizon $N \ll T_{N_\mathrm{tl}}$ in a receding horizon fashion.
Specifically, we design an MPC controller of the following form:
\begin{equation} \label{eq:MPC}
\begin{aligned}
    & J_\text{MPC} (\hat{\mathbf{x}}_{k}, \{\hat{s}^{pv}_{i|k}\}_{i=0}^{N}, \{\hat{v}^{pv}_{i|k}\}_{i=0}^{N}) = \\
    & \min_{\substack{u_{0:N-1|k}}} ~\, \sum_{i=0}^{N-1} \ell(\bar{\mathbf{x}}_{i|k}, ~ u_{i|k}) +\mathbb{E}_{n_{0:N-1}}\Bigr[ V_f(\hat{\mathbf{x}}_{N|k}) \Bigr]\\
    & ~~~~ \textnormal{s.t.,} \, ~~~~ \hat{\mathbf{x}}_{i+1|k} =  \mathbf{A} \hat{\mathbf{x}}_{i|k} + \mathbf{B} u_{i|k} + \mathbf{D}n_{i}, \\
    & \qquad \qquad ~ \bar{\mathbf{x}}_{i+1|k} =\mathbf{A} \bar{\mathbf{x}}_{i|k} + \mathbf{B} u_{i|k}, \\
    & \qquad \qquad ~ \hat{\mathbf{x}}_{0|k} = \bar{\mathbf{x}}_{0|k} = \hat{\mathbf{x}}_{k}, ~ n_i \sim p_n(n), \\
    & \qquad \qquad ~ \hat{\mathbf{x}}_{i|k} \in \mathcal{X}, ~ u_{i|k} \in \mathcal{U}, \\
    & \qquad \qquad ~ \hat{\mathbf{x}}_{i|k} \in \mathcal{C}(\hat{s}^{pv}_{i|k}, \hat{v}^{pv}_{i|k}),  ~ \forall n_i \in 2L\mathcal{W},\\
    & \qquad \qquad ~ \hat{\mathbf{x}}_{N|k} \in \mathcal{X}_f, ~ \forall n_i \in 2L\mathcal{W},\\
    & \qquad \qquad ~ \hat{\mathbf{x}}_{i|k} \in \mathcal{L}(s_\mathrm{tl}) \ominus \mathbf{D}\mathcal{W}, ~\text{if } c_{l^\star,k+i}=\text{red}, \\
    & \qquad \qquad ~ i \in \{0,\cdots,N-1\},
\end{aligned}
\end{equation}
where $V_f(\cdot)$ denotes a terminal cost, which represents the remaining energy-to-spend for the vehicle to cross the upcoming traffic light, $\bar{x} _{i|k}$ is a nominal state at predicted time step $k+i$, and $\mathcal{X}_f$ is a terminal set to satisfy the corresponding traffic light crossing time constraint in \eqref{eq:TL constraint}. Considering the localization uncertainty in \eqref{eq: gps error} and Proposition \ref{prop: delta s and lumped noise}, the constraint $\hat{\mathbf{x}}_{i|k} \in \mathcal{L}(s_\mathrm{tl}) \ominus \mathbf{D}\mathcal{W}$ implies ${\mathbf{x}}_{k+i} \in \mathcal{L}(s_\mathrm{tl})$. This implies that if the upcoming traffic light is red, the constraint prohibits the ego vehicle from crossing it.

We construct the terminal cost $V_f(\cdot)$ and the terminal set $\mathcal{X}_f$ in a data-driven way, which will be detailed in Sec. \ref{sec: data-driven Vf Xf}.
The expected stage cost in \eqref{eq:ftocp} becomes deterministic as described in \eqref{eq:MPC}. This is because $\pi_{i|k}(\hat{\mathbf{x}}_{i|k}) = u_{i|k}$, and the stage cost in \eqref{eq: energy cost regression} does not depend on the position. Thus, $\ell(\hat{\mathbf{x}}_{i|k}, u_{i|k}) = \ell(\bar{\mathbf{x}}_{i|k}, u_{i|k})$, making it deterministic.


To address (C3), we approximate the expected terminal cost with its tractable approximation introduced in \cite[Sec.III.B.1]{ejoalmpc}. 
To apply the approximation from \cite[Sec.III.B.1]{ejoalmpc}, we first reformulate the expectation with respect to the random variables $n_{0:N-1}$ into an expectation with respect to one random variable.
Using the dynamics in \eqref{eq:MPC}, we obtain the following equation:
\begin{equation} \label{eq: hatxN and barxN}
    \hat{\mathbf{x}}_{N|k} = \bar{\mathbf{x}}_{N|k} + \sum_{i=0}^{N-1} Dn_i = \bar{\mathbf{x}}_{N|k} + \mathbf{v}_N,
\end{equation}
where $\mathbf{v}_N = \sum_{i=0}^{N-1} Dn_i$. From \eqref{eq: lumped noise}, we can derive that $\mathbf{v}_N$ has the bounded support $2LN \mathcal{W}$.
Let $p_v(v)$ denote the probability density function of $\mathbf{v}_N$.
\textcolor{black}{In practice, differential GPS, providing centimeter-level positioning accuracy, can be utilized to calculate $\mathbf{v}_N$ and approximate its density function $p_v(v)$ prior to controller deployment. Note that once $p_v(v)$ has been identified, the differential GPS is no longer needed for the controller's operation.}
In summary, $\mathbf{v}_N$ is a random variable with the probability density function $p_v(v)$ and bounded support $2LN \mathcal{W}$ as follows:
\begin{equation} \label{eq: v noise}
\begin{aligned}
    & \mathbf{v}_N \sim p_v(v), ~ \mathbf{v}_N \in 2LN \mathcal{W}.
\end{aligned}
\end{equation}
Utilizing the $\mathbf{v}_N$, we can express the expected terminal cost in \eqref{eq:MPC} as follows:
\begin{equation} \label{eq: reform expectation}
    \mathbb{E}_{n_{0:N\text{-}1}}[ V_f(\hat{\mathbf{x}}_{N|k})] = \mathbb{E}_{\mathbf{v}_N}[ V_f(\hat{\mathbf{x}}_{N|k})] = \mathbb{E}_{\mathbf{v}_N}[ V_f(\bar{\mathbf{x}}_{N|k} + \mathbf{v}_N)],
\end{equation}
where the last equation follows from \eqref{eq: hatxN and barxN}.

In \cite[Sec.III.B.1]{ejoalmpc}, the expected terminal cost in \eqref{eq: reform expectation} is approximated with a convex combination of the $V_f(\cdot)$ values at multiple number of points as follows:
\begin{equation} \label{eq: discrete sum of expectation}
    \mathbb{E}_{\mathbf{v}_N}\Bigr[ V_f(\bar{\mathbf{x}}_{N|k} + \mathbf{v}_N) \Bigr] \simeq \sum_{m=1}^M p_m V_f(\bar{\mathbf{x}}_{N|k} + \mathbf{v}_N^m),
\end{equation}
where $\mathbf{v}_N^m$ denotes $m$-th discretized noises, $M$ represents the number of the discretized noises, and $p_m$ is a convex coefficient of each $\mathbf{v}_N^m$.
The discretized noises $\mathbf{v}_N^m, ~m=\{1,\cdots,M\}$ consist of the recorded noises and the vertices of the bounded support $2LN \mathcal{W}$ as shown in Fig. \ref{fig:disturbance set grid}. The calculation of $p_m$ for each $\mathbf{v}_N^m$ is conducted before solving the MPC (or offline). The details of the calculation can be found in \cite[Appendix.A.1-2]{ejoalmpc}.
\vspace{-1.0em}
\begin{figure}[ht]
\begin{center}
\includegraphics[width=0.85\linewidth,keepaspectratio]{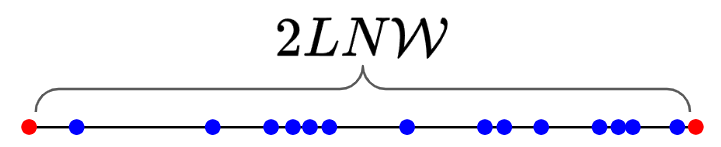}
\vspace{-1.0em}
\caption{The black line is the bounded support $2LN\mathcal{W}$. All dots are the discretized noises. The blue dots denote the recorded noises while the red dots denote vertices of the bounded support $2LN\mathcal{W}$.}
\label{fig:disturbance set grid}
\end{center}
\end{figure}
\vspace{-1.0em}


\section{Data-driven Design of $V_f(\cdot)$ and $\mathcal{X}_f$} \label{sec: data-driven Vf Xf}
In this section, we introduce a data-driven way to construct the terminal cost function $V_f(\cdot)$ and the terminal set $\mathcal{X}_f$, which is proposed in \cite{ejoalmpc}.
The terminal cost function $V_f(\cdot)$ represents the remaining energy-to-spend until the vehicle crosses the upcoming traffic light.
The terminal set $\mathcal{X}_f$ is designed to satisfy the corresponding traffic light crossing time constraints in \eqref{eq:TL constraint}, which impose the vehicle state to cross the upcoming traffic light within the predetermined time step.





This section is structured as follows. 
First, we introduce the dataset generation process following \cite[Sec. IV. A \& Appendix. B]{ejoalmpc}. 
Second, we present the data-driven construction of $V_f(\cdot)$, following \cite[Sec. IV. C]{ejoalmpc}. 
Third, we explain the data-driven construction of $\mathcal{X}_f$ following \cite[Sec. IV. C]{ejoalmpc}.

\subsection{Data: State-Input pairs} \label{sec: data}
We generate the dataset through two processes, initialization and augmentation. The data set is generated through simulation in this paper at Sec. \ref{sec: single road segment}. However, it is worth mentioning that data generation is not limited to simulation; it can also be achieved through closed-loop testing.

\subsubsection{Initialization}
To initialize the dataset, we simulate a simple scenario rather than a complex urban road driving scenario.
In this scenario, the vehicle with zero initial speed is on a single road segment with just one traffic light, located 200 meters away, i.e., $s_\mathrm{tl}^0=200$. The goal for the vehicle is to cross the traffic light within a $20$-second.

We design an MPC for a cruise control algorithm as in \cite[Sec. IV. C]{ejoaiv}, which satisfies the constraints \eqref{eq:ftocp} and tracks a reference speed $15$m$/$s, the speed limit of the road. 
Let $T$ denote the time step that the controlled vehicle crosses the traffic light located 200 meters away. We collect the optimal input sequence $\bar{u}_{0:T}^\star$ and the associated state $\bar{\mathbf{x}}_{0:T}^\star$. Based on these, we construct the initial data set as follows:

\begin{equation} \label{eq: initial dataset}
\begin{split}
    & \mathbf{X}^0 = \begin{bmatrix} \bar{\mathbf{x}}^\star_0 - \mathbf{D}s_\mathrm{tl}^0 &
    \bar{\mathbf{x}}^\star_1 - \mathbf{D}s_\mathrm{tl}^0 &
    \cdots &
    \bar{\mathbf{x}}^\star_{T} - \mathbf{D}s_\mathrm{tl}^0
    \end{bmatrix}, \\
    & \mathbf{U}^0 = \begin{bmatrix} \bar{u}_0^\star &
    \bar{u}_1^\star &
    \cdots &
    \bar{u}_{T}^\star \end{bmatrix},
\end{split}
\end{equation}
where the superscript $^0$ denotes an initialization, and $\mathbf{X}^0$ and $\mathbf{U}^0$ represent an initial state and an initial input data matrices, respectively.
Note that we transform the coordinate system of the state by subtracting the position of the upcoming traffic light $s_\mathrm{tl}^0=200$, with the element now indicating the remaining distance to the traffic light.

\begin{remark}
\textcolor{black}{To address a potential limitation due to initialization with data collected within a fixed setting (200 meters and 20 seconds, 15m$/$s), the initialization parameters can be modified to match the specific requirements of different scenarios. For instance, in urban driving scenarios where traffic lights are located 1000 meters away, the initialization parameter can be adjusted from 200 meters to over 1000 meters. Moreover, the initialization step can be conducted for multiple times with different reference speeds.}
\end{remark}

\subsubsection{Update}
We recursively augment the data set starting from the initial data set \eqref{eq: initial dataset}.
We describe one iteration of the data augmentation process below.
Using the provided data, we construct the terminal cost function $V_f(\cdot)$ as per Sec. \ref{sec: learning V function} and the terminal set $\mathcal{X}_f$ as per Sec. \ref{sec: learning constraints}. 
With these components in place, we construct the proposed MPC by updating the terminal cost $V_f(\cdot)$ and the terminal constraint using $\mathcal{X}_f$, as detailed in Section \ref{sec: Data-driven MPC}.
Subsequently, we control the vehicle with this updated MPC for urban road driving scenarios and record the closed-loop state and the corresponding optimal input at every time step. Finally, we augment the data set as follows:
\begin{equation} \label{eq: j-th dataset}
\begin{split}
    &\mathbf{X}^j = \begin{bmatrix} {\mathbf{X}^{j-1}} &
    \mathbf{x}_0 - \mathbf{D}s_{\mathrm{tl},l^\star(\mathbf{x}_0)} &
    \cdots &
    \mathbf{x}_{N^{a}_\text{data}} - \mathbf{D}s_{\mathrm{tl},l^\star(\mathbf{x}_{N^{a}_\text{data}})} \end{bmatrix}, \\
    & \mathbf{U}^j = \begin{bmatrix} {\mathbf{U}^{j-1}} &
    u_0 &
    \cdots &
    u_{N^{a}_\text{data}} \end{bmatrix},
\end{split}
\end{equation}
where the superscript $j$ represents the iteration of the data augmentation, $N^{a}_\text{data}$ represents the number of augmenting data points, and $l^\star(\mathbf{x})$ is the optimal solution of \eqref{eq: s_tl(x)}.
Let $j_\text{max}$ denote the maximum update iteration.
Throughout the paper, we define $\mathbf{X} = \mathbf{X}^{j_\text{max}}$ and $\mathbf{U} = \mathbf{U}^{j_\text{max}}$ for brevity.

\subsection{Learning the terminal cost $V_f(\cdot)$} \label{sec: learning V function}
The terminal cost function $V_f(\cdot)$ represents the remaining energy-to-spend until the vehicle crosses the upcoming traffic light.
Calculating the terminal cost function $V_f(\cdot)$ exactly presents significant challenges due to the dynamic and uncertain nature of traffic environments. 

We propose a practical approach to approximate $V_f(\cdot)$ by considering the remaining energy-to-spend under certain simplifying assumptions: (1) the time to cross the upcoming traffic light is unconstrained, and (2) there is no front vehicle.
This simplification allows us to focus on the intrinsic energy dynamics of the vehicle itself, decoupled from external factors such as traffic flow and vehicle interactions.

\textcolor{black}{Previous research with extensive vehicle tests \cite{bae2019real, ard2023VILCAV} shows that tracking the speed reference for the entire route under free-flow traffic conditions can improve energy efficiency even with moderate traffic congestion. This finding also aligns with our results. We empirically find that the terminal cost function $V_f(\cdot)$ with simplifying assumptions is a good heuristic for urban road driving aimed at minimizing energy consumption.}


For the rest of the paper, we will use the term "cost-to-go" to refer to the remaining energy-to-spend.
The terminal cost function $V_f(\cdot)$ will be written as a convex combination of known cost-to-go values of the collected data points.

\subsubsection{Initialization}
There are three stages of initializing the terminal cost $V_f(\cdot)$, i.e., computing $V_f^0(\cdot)$ where the superscript $^0$ represents the initialization.

First, let $\mathcal{O}$ denote the region behind the upcoming traffic light located at $s_\text{tl}^0=200$, i.e., $\mathcal{O}=\{\mathbf{x} ~|~ \begin{bmatrix} 1 & 0 \end{bmatrix} \mathbf{x} \geq s_\mathrm{tl}^0, ~\mathbf{x} \in \mathcal{X}\}$. We assign the cost-to-go value to the vertices of $\mathcal{O}$.
Let $v^\mathcal{O}_i$ denote the $i$-th vertex of $\mathcal{O}$, and let $l_\mathcal{O}$ represent the number of vertices of $\mathcal{O}$.
Reminding that $\mathcal{O}$ is a polytope, we can summarize the cost-to-go values for the points within $\mathcal{O}$ as the cost-to-go values of the vertices of $\mathcal{O}$:
\begin{equation}
    \begin{aligned}
        & \mathbf{V}_\mathcal{O} = \begin{bmatrix} v^\mathcal{O}_0 - \mathbf{D}s_\mathrm{tl}^0 & 
        v^\mathcal{O}_1 - \mathbf{D}s_\mathrm{tl}^0 & 
        \cdots & 
        v^\mathcal{O}_{l_\mathcal{O}} - \mathbf{D}s_\mathrm{tl}^0\end{bmatrix}, \\
        & ~\mathbf{J}_\mathcal{O} = \begin{bmatrix} 0 & 0 & \cdots & 0 \end{bmatrix},
    \end{aligned}
\end{equation}
where $\mathbf{V}_\mathcal{O}$ is a matrix where each element is the vertices of $\mathcal{O}$ after undergoing a coordinate system transformation.
In this transformation, each element now denotes the remaining distance to the traffic light.
$\mathbf{J}_\mathcal{O}$ is a matrix where each element corresponds to the remaining energy-to-spend at the corresponding element of $\mathbf{V}_\mathcal{O}$.
By definition, the cost-to-go value is zero at every point in the set $\mathcal{O}$; thus, all elements of $\mathbf{J}_\mathcal{O}$ are zero.

Second, given $\mathbf{X}^0$ and $\mathbf{U}^0$ \eqref{eq: initial dataset}, we calculate the cost-to-go for all data points until the vehicle crosses the upcoming traffic light.
Let $[\mathbf{X}^0]_i$ and $[\mathbf{U}^0]_i$ denote $i$-th data point. 
We compute the cost-to-go value of the $i$-th data point as follows:
\begin{equation}
\begin{split}
    & J_T^0 = 0 ~(\because [\mathbf{X}^0]_T \in \mathcal{O}), \\
    & J_k^0 = \ell([\mathbf{X}^0]_k, ~[\mathbf{U}^0]_k) + J_{k+1}^0 
\end{split}
\end{equation}
We then stack these values for all data points to create an initial cost-to-go data vector, represented as $\mathbf{J}^{0}$ as follows:
\begin{equation}
    \mathbf{J}^{0} = \begin{bmatrix} J_0^0, J_1^0, \cdots, J_T^0 \end{bmatrix},
\end{equation}
where each element of $\mathbf{J}^{0}$ denote the cost-to-go value of the corresponding data point in $[\mathbf{X}^0]$ and $[\mathbf{U}^0]$.

Finally, the terminal cost $V_f^0(\cdot)$ is initialized as a convex combination of the calculated cost-to-go values. 
\begin{equation} \label{eq: initial value function}
\begin{split}
    & V_f^0(\mathbf{x}) = \min_{\bm{\lambda}^0, \bm{\lambda}^\mathcal{O}} ~\mathbf{J}^0 \bm{\lambda}^0 + \mathbf{J}_\mathcal{O} \bm{\lambda}^\mathcal{O} \\
    & \qquad ~~~~~~ \, \textnormal{s.t.,} ~~~ \mathbf{X}^0 \bm{\lambda}^0 + \mathbf{V}_\mathcal{O} \bm{\lambda}^\mathcal{O} = \mathbf{x} - \mathbf{D}\mathbf{s}_\mathrm{tl}, \\
    & \qquad \qquad ~~~~~~~~ \bm{\lambda}^0 \geq \bm{0}, \,\, \bm{1}^\top \bm{\lambda}^0 = 1,\\
    & \qquad \qquad ~~~~~~~~ \bm{\lambda}^\mathcal{O} \geq \bm{0}, \,\, \bm{1}^\top \bm{\lambda}^\mathcal{O} = 1.
\end{split}
\end{equation}
Note that we take a convex combination of the rows in $\mathbf{X}^0$ and $\mathbf{V}_\mathcal{O}$ in the first constraint. 
$V_f^0(\mathbf{x})$ is a convex, piecewise affine function as it is an optimal objective function of the multi-parametric linear program \cite{borrelli2017predictive}. 

\subsubsection{Update}
After the state and input data are updated as in \eqref{eq: j-th dataset}, we update the cost-to-go vector $[\mathbf{J}^{j}]$. The $i$-th element of the cost-to-go vector $[\mathbf{J}^{j}]_i$ is updated as follows:
\begin{equation} \label{eq: cost data}
    \begin{aligned}
        & [\mathbf{J}^{j}]_i =  \ell([\mathbf{X}^{j}]_i, [\mathbf{U}^{j}]_i) \\
        & ~~~~~~~~~~ + \sum_{m=1}^M p_m V_f^{j-1}(\mathbf{A}[\mathbf{X}^{j}]_i + \mathbf{B}[\mathbf{U}^{j}]_i + \mathbf{v}_N^m).
    \end{aligned}
\end{equation}
We update $V_f^{j}(\cdot)$ by solving the following optimization problem using \eqref{eq: j-th dataset} and \eqref{eq: cost data}. 
\begin{equation} \label{eq: value function update}
\begin{split}
    & V_f^{j}(\mathbf{x}) = \min_{\bm{\lambda}^{j}, \bm{\lambda}^\mathcal{O}}~ \mathbf{J}^{j} \bm{\lambda}^{j} + \mathbf{J}_\mathcal{O} \bm{\lambda}^\mathcal{O} \\
    & \qquad \qquad ~~ \textnormal{s.t.,} ~~~ \mathbf{X}^{j} \bm{\lambda}^{j} + \mathbf{V}_\mathcal{O} \bm{\lambda}^\mathcal{O} = \mathbf{x} - \mathbf{D}\mathbf{s}_\mathrm{tl}, \\
    & \qquad \qquad \qquad ~~~ \, \bm{\lambda}^{j} \geq \bm{0}, ~ \bm{1}^\top \bm{\lambda}^{j} = 1, \\
    & \qquad \qquad \qquad ~~~ \, \bm{\lambda}^\mathcal{O} \geq \bm{0}, \,\, \bm{1}^\top \bm{\lambda}^\mathcal{O} = 1.
\end{split}
\end{equation}
$V_f^{j}(x)$ is a convex, piecewise affine function as it is an optimal objective function of the multi-parametric linear program \cite{borrelli2017predictive}.

Remind that $j_\text{max}$ denotes the maximum update iteration.
The terminal cost function can be obtained by $V_f(\cdot) = V_f^{j_\text{max}}(\cdot)$.
\begin{remark}
The function $V^j(x)$ does not have to be calculated before computing the optimal input of the proposed MPC because \eqref{eq: value function update} can be incorporated into the proposed MPC and form a single unified optimization problem.  
\end{remark}

\subsection{Learning the terminal set $\mathcal{X}_f$} \label{sec: learning constraints}
The terminal constraint $\hat{\mathbf{x}}_{N|k} \in \mathcal{X}_f$ in \eqref{eq:MPC} is designed to ensure that the vehicle crosses the traffic light within the predetermined time step given by the high-level module while satisfying the collision avoidance constraints introduced in \eqref{eq:MPC}.
We design two sets in a data-driven way, and the intersection of these two sets will define the terminal set $\mathcal{X}_f$.

We adopt the notion of \textit{Robust Controllable Set} to design these two sets.
A set is $N$-step \textit{Robust Controllable}, if all states belonging to the set can be robustly driven, through a time-varying control law, to the target set in $N$-step \cite[Def.10.18]{borrelli2017predictive}.
We design two \textit{Robust Controllable Sets} where all estimated states of the system \eqref{eq: observer equation} in each set can be robustly steered to each target set against additive lumped noise $n$ \eqref{eq: lumped noise} while satisfying the collision avoidance constraints:
\begin{enumerate} 
    \item $\mathcal{S}_{t_\text{red}}$: $t_\text{red}$-step Robust Controllable Set where the target set is $\hat{\mathcal{G}}_s = \{\mathbf{x} ~|~ \begin{bmatrix} 1 & 0 \end{bmatrix} \mathbf{x} \leq s_\mathrm{tl} - 2Lw_\text{max}\}$.
    \item $\mathcal{P}_{t_\text{green}}$: $t_\text{green}$-step Robust Controllable Set where the target set is $\hat{\mathcal{G}}_p = \{\mathbf{x} ~|~ \begin{bmatrix} 1 & 0 \end{bmatrix} \mathbf{x} \geq s_\mathrm{tl} - 2Lw_\text{min}\}$,
\end{enumerate}
where $w_\text{min}$ and $w_\text{max}$ denote the extreme of the segment $\mathcal{W}$ defined in \eqref{eq: gps error}.
In Fig. \ref{fig:definition of tred tgreen}, $t_\text{red}$ and $t_\text{green}$ are depicted. Remind that $T_l$ denotes the predetermined time to cross the $l$-th traffic light, and $l^\star$ denotes the index of the nearest upcoming traffic light. Thus, $T_{l^\star}$ represents the predetermined time to cross the nearest upcoming traffic light.
\begin{remark}
Let $\mathcal{G}_s = \{\mathbf{x} ~|~ \begin{bmatrix} 1 & 0 \end{bmatrix} \mathbf{x} \leq s_\mathrm{tl}(\mathbf{x})\}$ and $\mathcal{G}_p = \{\mathbf{x} ~|~ \begin{bmatrix} 1 & 0 \end{bmatrix} \mathbf{x} \geq s_\mathrm{tl}(\mathbf{x})\}$. These sets denote the regions before and after the upcoming traffic light, respectively.
Note that the target sets $\hat{\mathcal{G}}_s$ and $\hat{\mathcal{G}}_p$ are obtained by $\hat{\mathcal{G}}_s = \mathcal{G}_s \ominus \mathbf{D}\mathcal{W}$ and $\hat{\mathcal{G}}_p = \mathcal{G}_p \ominus \mathbf{D}\mathcal{W}$.
From Proposition \ref{prop: delta s and lumped noise}, $\mathbf{x} - \hat{\mathbf{x}} \in \mathbf{D}\mathcal{W}$. 
Thus, $\hat{\mathbf{x}} \in \hat{\mathcal{G}}_s$ and $\hat{\mathbf{x}} \in \hat{\mathcal{G}}_p$ imply ${\mathbf{x}} \in {\mathcal{G}}_s$ and ${\mathbf{x}} \in {\mathcal{G}}_p$, respectively.
\end{remark}
\begin{figure}[ht]
\vspace{-1.0em}
\begin{center}
\includegraphics[width=0.75\linewidth,keepaspectratio]{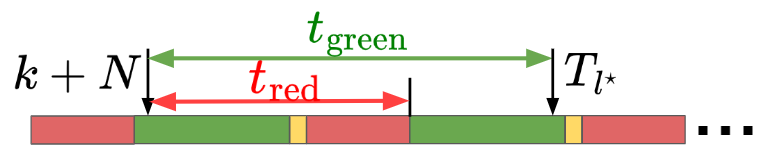}
\vspace{-1.0em}
\caption{Defintion of $t_\text{red}$ and $t_\text{green}$}
\label{fig:definition of tred tgreen}
\end{center}
\end{figure}
\vspace{-2.0em}
\begin{figure}[ht]
\begin{center}
\includegraphics[width=0.75\linewidth,keepaspectratio]{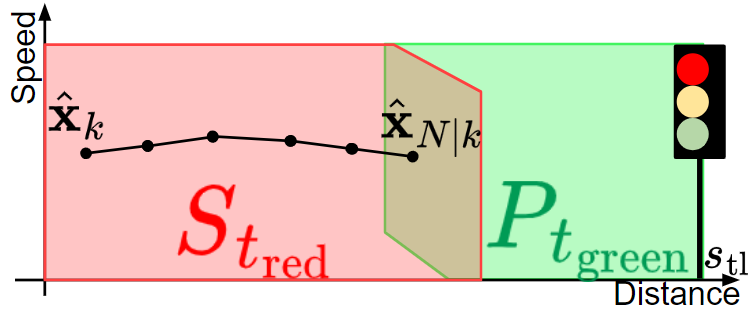}
\vspace{-1.0em}
\caption{Illustration of the terminal constraints $\mathcal{S}_{t_\text{red}}$ and $\mathcal{P}_{t_\text{green}}$: the constraint $\hat{\mathbf{x}}_{N|k} \in \mathcal{S}_{t_\text{red}}$ is to ensure that the vehicle stays behind the traffic light until $t_\text{red}+N$ steps, while the constraint $\hat{\mathbf{x}}_{N|k} \in \mathcal{P}_{t_\text{green}}$ is to ensure that the vehicle passes the traffic light within $t_\text{green}+N$ steps.}
\label{fig:planning_strategy}
\end{center}
\end{figure}

The terminal set $\mathcal{X}_f$ is defined as the intersection of these two sets, $\mathcal{S}_{t_\text{red}} \medcap \mathcal{P}_{t_\text{green}}$.
Let $\hat{\mathbf{x}}_{N|k}$ denotes a predicted estimated state at time step $k$+$N$ calculated from the current estimate $\hat{\mathbf{x}}_{k}$, i.e., a terminal state.
Then, the terminal constraint, $\hat{\mathbf{x}}_{N|k} \in \mathcal{X}_f = \mathcal{S}_{t_\text{red}} \medcap \mathcal{P}_{t_\text{green}}$, enforces not to pass the traffic light before it turns green while guaranteeing that the vehicle crosses the traffic light within the time predetermined  by the high-level module.
Fig. \ref{fig:planning_strategy} illustrates the terminal constraints. 

In this paper, the sets $\mathcal{S}_{t_\text{red}}$ and $\mathcal{P}_{t_\text{green}}$ are calculated in a data-driven way by Algorithm \ref{alg: RNTW} with $\hat{\mathcal{G}}_s$ and $\hat{\mathcal{G}}_p$ as the respective target convex sets. 
The corresponding outputs of Algorithm \ref{alg: RNTW} are assigned to $\mathcal{S}_{t_\text{red}}$ and $\mathcal{P}_{t_\text{green}}$.

\begin{algorithm} \label{alg: RNTW}
    \SetKwInOut{KwIn}{Input}
    \SetKwInOut{KwOut}{Output}
    \KwIn{Data matrices $\mathbf{X}, \mathbf{U}$, \\ Remaining steps $t$, Target convex set $\mathcal{G}$, \\ Predicted behavior of the front vehicle \\ $\{\hat{s}^{pv}_{i|k}\}_{i=N}^{t+N}$, $\{\hat{v}^{pv}_{i|k}\}_{i=N}^{t+N}$} 
    \KwOut{$t$-step \textit{Robust Controllable Set} $\mathcal{R}_t$} 

    $\mathcal{R}_0 \leftarrow \mathcal{G} \medcap \mathcal{C}(\hat{s}^{pv}_{t+N|k}, \hat{v}^{pv}_{t+N|k})$
    
    
    \For{$p \leftarrow 0$ \KwTo $t-1$}{
        $\mathbb{X} \leftarrow \emptyset$
        
        \For{$q \leftarrow 1$ \KwTo $N_{\mathbf{X}}$}{
            \If{$\mathbf{A}[\mathbf{X}]_q + \mathbf{B}[\mathbf{U}]_q \in \mathcal{R}_{p} \ominus 2L \mathbf{D} \mathcal{W}$}{
            \If{$[\mathbf{X}]_q \in \mathcal{C}(\hat{s}^{pv}_{t+N-p-1|k}, \hat{v}^{pv}_{t+N-p-1|k})$}
            {
                $\mathbb{X} \leftarrow \mathbb{X} \medcup \{[\mathbf{X}]_q\}$
                
                }
             }
        }
        $\mathcal{R}_{p+1} \leftarrow \text{conv}(\mathbb{X})$ \\
    }
    \KwRet{$\mathcal{R}_t$} 
    \caption{Data-driven computation of \\ $~~~~~~~~~~~~~~~~~~ t$-step \textit{Robust Controllable Set} $\mathcal{R}_t$}
\end{algorithm}
$\mathcal{R}_{i}$ represents $i$-step robust controllable set, which is initialized with to the target set $\mathcal{G}$ at line 2.
In line 4, $N_{\mathbf{X}}$ denotes the number of the column vectors of $\mathbf{X}$.
Line $5$ at the iteration $i$ checks whether the state $[\mathbf{X}]_j$ of the system \eqref{eq: observer equation} can be robustly steered to $\mathcal{R}_{i-1}$. Then, at line $6$, the collision avoidance satisfaction is checked at the state $[\mathbf{X}]_j$. If both conditions are satisfied, the state $[\mathbf{X}]_j$ is included in a set $\mathbb{X}$. At the end of the iteration $i$ (line 11), $\mathcal{R}_i$ is set as the convex hull of the set $\mathbb{X}$.
By the following theorem, the output of the Algorithm \ref{alg: RNTW} is a $t$-step robust controllable set for a given target convex set $\mathcal{G}$. 
\begin{theorem} \label{thm1}
    The output of Algorithm \ref{alg: RNTW} $\mathcal{R}_t$ is $t$-step robust controllable set of the system \eqref{eq: observer equation} perturbed by the noise \eqref{eq: lumped noise} for a given target convex set $\mathcal{G}$ subject to the constraints \eqref{eq:constraints} and the collision avoidance constraints given the predicted behavior of the front vehicle $\{\hat{s}^{pv}_{i|k}\}_{i=N}^{t+N}$ and $\{\hat{v}^{pv}_{i|k}\}_{i=N}^{t+N}$.  
\end{theorem}
\begin{proof}
    See Appendix \ref{app. thm}.
\end{proof}

\section{Controller} \label{sec: Data-driven MPC}
We implement two MPC controllers: 
\begin{itemize}
    \item If the remaining time exceeds $k+N \leq T_{l^\star}$, we employ the standard MPC controller introduced in \eqref{eq:mpc reform}, which is the primary choice given its applicability to most scenarios.
    \item Conversely, if the remaining time is less than $k+N > T_{l^\star}$, we opt for the simplified version of the standard MPC controller introduced in \eqref{eq: mpc simplified}. 
\end{itemize}
Recall that $T_{l^\star}$ is the predetermined time to cross the nearest upcoming traffic light as illustrated in Fig. \ref{fig:definition of tred tgreen}.

First, we introduce the standard MPC controller with prediction horizon $N$.
We incorporate the terminal cost function $V_f(\mathbf{x})$ in \eqref{eq: value function update} and the terminal set $\mathcal{X}_f=\mathcal{S}_{t_\text{red}} \medcap \mathcal{P}_{t_\text{green}}$ into the MPC controller in \eqref{eq:MPC}.
Additionally, we reformulate the constraints on estimate states to those on nominal states\footnote{See Appendix \ref{app. constraint reform} for the details of the constraint reformulation.}.
Specifically, we design an  MPC controller as follows:
\begin{equation} \label{eq:mpc reform}
\begin{aligned}
    & \Tilde{J}_\text{MPC} (\hat{\mathbf{x}}_{k}, \{\hat{s}^{pv}_{i|k}\}_{i=0}^{N}, \{\hat{v}^{pv}_{i|k}\}_{i=0}^{N}) = \\
    & \min_{\substack{u_{0:N-1|k}, \\ \mathbf{s}_f}}  \sum_{i=0}^{N-1} \ell(\bar{\mathbf{x}}_{i|k}, ~ u_{i|k}) +\sum_{m=1}^M p_m V_f(\bar{\mathbf{x}}_{N|k} + \mathbf{v}_N^m) + M_f \mathbf{s}_f\\
    & ~~~~ \textnormal{s.t.,} \, ~~ \bar{\mathbf{x}}_{i+1|k} =\mathbf{A} \bar{\mathbf{x}}_{i|k} + \mathbf{B} u_{i|k} \\
    & \qquad ~~~~~ \bar{\mathbf{x}}_{0|k} = \hat{\mathbf{x}}_{k}, \\
    & \qquad ~~~~~  \bar{\mathbf{x}}_{i|k} \in \mathcal{X}, ~ u_{i|k} \in \mathcal{U},\\
    & \qquad ~~~~~  \bar{\mathbf{x}}_{i|k} \in \mathcal{C}(\hat{s}^{pv}_{i|k}, \hat{v}^{pv}_{i|k}) \ominus 2Li\mathbf{D}\mathcal{W} \\
    & \qquad ~~~~~  \bar{\mathbf{x}}_{N|k} \in (\mathcal{S}_{t_\text{red}} \medcap \mathcal{P}_{t_\text{green}}) \ominus 2LN\mathbf{D}\mathcal{W} \oplus \mathbf{s}_f,\\
    & \qquad ~~~~~  \bar{\mathbf{x}}_{i|k} \in \mathcal{L}(s_\mathrm{tl}) \ominus (2L+1)\mathbf{D}\mathcal{W}, \text{if} \,  c_{l^\star,k+i}=\text{red}, \\
    & \qquad ~~~~~  \mathbf{s}_f \geq 0, \\
    & \qquad ~~~~~  i \in \{0,\cdots,N-1\},
\end{aligned}
\end{equation}
where $M_f$ is a sufficiently large number, which is set to $10000$ in this paper, and $\mathbf{s}_f$ denotes a slack variable, which is required to ensure feasibility when the data to construct the terminal set is insufficient, or the surrounding environment abruptly changes, such as the front vehicle abruptly decelerating or another vehicle cutting into the ego vehicle's lane.
Note that the objective function in \eqref{eq:mpc reform} is identical to that in \eqref{eq:MPC} since the stage cost in \eqref{eq: energy cost regression} does not depend on the position and thus $\ell(\hat{\mathbf{x}}_k, u_k) = \ell(\bar{\mathbf{x}}_k, u_k)$.
The minimization problem \eqref{eq: initial value function} can be incorporated into the problem \eqref{eq:mpc reform} and form a single unified optimization problem\footnote{See Appendix \ref{app: integration MPC} for the formulation.}.

Second, we introduce the simplified version of the standard MPC controller \eqref{eq:mpc reform} with shrinking prediction horizon $N_s = T_{l^\star} - k$. Specifically, we design an  MPC controller as:
\begin{equation} \label{eq: mpc simplified}
\begin{aligned}
    & \Tilde{J}_\text{MPC}^\text{simplified} (\hat{\mathbf{x}}_{k}, \{\hat{s}^{pv}_{i|k}\}_{i=0}^{N_s}, \{\hat{v}^{pv}_{i|k}\}_{i=0}^{N_s}) = \\
    & \min_{\substack{u_{0:N_s-1|k},\\ \mathbf{s}_f}} \sum_{i=0}^{N_s-1} \ell(\bar{\mathbf{x}}_{i|k}, ~ u_{i|k}) +  M_f \mathbf{s}_f\\
    & ~~~~ \textnormal{s.t.,} \, ~~ \bar{\mathbf{x}}_{i+1|k} =\mathbf{A} \bar{\mathbf{x}}_{i|k} + \mathbf{B} u_{i|k} \\
    & \qquad ~~~~~  \bar{\mathbf{x}}_{0|k} = \hat{\mathbf{x}}_{k}, \\
    & \qquad ~~~~~  \bar{\mathbf{x}}_{i|k} \in \mathcal{X}, ~ u_{i|k} \in \mathcal{U},\\
    & \qquad ~~~~~  \bar{\mathbf{x}}_{i|k} \in \mathcal{C}(\hat{s}^{pv}_{i|k}, \hat{v}^{pv}_{i|k}) \ominus 2Li\mathbf{D}\mathcal{W} \\
    & \qquad ~~~~~  \bar{\mathbf{x}}_{N_s|k} \in \hat{\mathcal{G}}_p \ominus 2LN_s\mathbf{D}\mathcal{W}\oplus \mathbf{s}_f,\\
    & \qquad ~~~~~  \bar{\mathbf{x}}_{i|k} \in \mathcal{L}(s_\mathrm{tl}) \ominus (2Li+1)\mathbf{D}\mathcal{W}, \text{if} \, c_{l^\star,k+i}=\text{red}, \\
    & \qquad ~~~~~  \mathbf{s}_f \geq 0 \\
    & \qquad ~~~~~  i \in \{0,\cdots,N_s-1\}.
\end{aligned}
\end{equation}
Remind that ${\mathcal{G}}_p$ denote the region after the upcoming traffic light, and $\hat{\mathcal{G}}_p = {\mathcal{G}}_p \ominus \mathbf{D}\mathcal{W}$ as defined in Sec. \ref{sec: learning constraints}. 
If $\mathbf{s}_f = 0$, the terminal constraint of the \eqref{eq: mpc simplified}, i.e., $\bar{\mathbf{x}}_{N_s|k} \in \hat{\mathcal{G}}_p \ominus 2LN_s\mathbf{D}\mathcal{W}\oplus \mathbf{s}_f$, imply that $\hat{\mathbf{x}}_{N_s|k} \in \hat{\mathcal{G}}_p$ and thus ${\mathbf{x}}_{N_s|k} \in {\mathcal{G}}_p$. By definition, the cost-to-go value for all $\mathbf{x}_{N_s|k} \in {\mathcal{G}}_p$ is zero. Therefore, all cost functions and constraints related to calculating the terminal cost $V_f(\cdot)$ are removed from \eqref{eq:mpc reform}.

In summary, if $k+N \leq T_{l^\star}$, the MPC controller in \eqref{eq:mpc reform} is employed; otherwise, the MPC controller in \eqref{eq: mpc simplified} is employed.
The resulting control policy is:
\begin{equation} \label{eq: control policy final}
    \pi(\hat{x}_k)= 
    \begin{cases}
        u^\star_{0|k} \text{ from \eqref{eq:mpc reform}} ,& \text{if } k+N \leq T_{l^\star},\\
        u^\star_{0|k} \text{ from \eqref{eq: mpc simplified}},              & \text{otherwise}.
    \end{cases}
\end{equation}

\begin{theorem} \label{thm3}
     Consider the closed-loop system \eqref{eq:system} controlled by the control policy \eqref{eq: control policy final}.
     If the optimal slack variable $\mathbf{s}_f^\star = 0$ at time step $k$, then there exists a control input sequence that steers the system \eqref{eq:system} to the region after the upcoming traffic light within the corresponding predetermined time, while satisfying the constraints \eqref{eq: collision avoidance constraint}-\eqref{eq: light rule constraint}.
\end{theorem}
\begin{proof}
    See Appendix \ref{app. thm3}.
\end{proof}
\begin{remark}
Theorem \ref{thm3} describes that if $\mathbf{s}_f^\star \neq 0$, we cannot ensure that there is a control sequence that steers the system \eqref{eq:system} after the upcoming traffic light within the predetermined time. This could happen in practice if the data is insufficient and/or the predicted speed is overestimated by the actual speed of the front vehicle. It is also possible that the predetermined time computed by the high-level green wave search module is inaccurate. Thus, in this case, we increase the remaining time step $T_{l^\star(\mathbf{x})}$ illustrated in Fig. \ref{fig:definition of tred tgreen} and relax the problem, which then solves the problem of crossing the upcoming traffic light with a larger remaining time step.
\end{remark}

\section{Validation} \label{sec: experiment}
\subsection{Energy Consumption Model}
We collected the energy consumption data of our electric test vehicle and solved the regression problem \eqref{eq: energy cost regression}. The energy consumption measurement data used for the regression was collected from 3.9 km of city driving scenarios in a test track. 
The regression results show that the error of the total energy consumption is below $1 \%$. 

We validated our regression model in \eqref{eq: energy cost regression} using an additional dataset not employed during the regression process. This test dataset is crucial for assessing the model's generalizability. We conducted 15 additional tests. For each trial, we calculated the percentage error in total energy consumption between measurements and the regression model \eqref{eq: energy cost regression}. The results are described in Table. \ref{table: energy eval2}, which indicates that the mean error is around $1 \%$, and the worst-case error is $6.3 \%$.
\vspace{-0.5em}
\begin{table}[h]
\centering
\caption{Error of the cumulative energy consumption (\%)}
\vspace{-2mm}
\label{table: energy eval2}
\begin{tabular}{ ||l||c||c||} 
 \hline
 Mean[\%]&Standard deviation[\%]&Worst case[\%]\\
 \hline
 $-1.06\%$&$3.41 \%$&$-6.3\%$\\
 \hline
\end{tabular}
\end{table}
\vspace{-0.5em}

\subsection{Baseline} \label{sec: baseline}
Two baseline algorithms were considered in this study. 
The first baseline algorithm is a cruise controller in \cite{ejoaiv}. 
The second baseline algorithm is the algorithm introduced in \cite{bae2019VIL}. 
The major differences between the algorithm in \cite{bae2019VIL} and the proposed algorithm are:
\begin{itemize}
    \item The algorithm in \cite{bae2019VIL} disregards localization uncertainty, while the proposed algorithm addresses this uncertainty.
    \item The algorithm in \cite{bae2019VIL} needs to compute speed references for the entire route, while the proposed algorithm calculates speed references for a short prediction horizon, requiring less computational power.
    \item The algorithm in \cite{bae2019VIL} has a hierarchical structure: a speed planner and an additional speed tracking controller, resulting in the latency and compounding error issue. In contrast, the proposed algorithm is a single MPC controller with longitudinal acceleration input.
\end{itemize}
Note that for three algorithms, there is an actuator-level controller that converts longitudinal acceleration commands into electric motor torque inputs. 
\textcolor{black}{This controller was developed by the manufacturer and is utilized to control the electric motors in the vehicle when its cruise controller is operating.}

\subsection{Localization Uncertainty and Controller Parameters}
Throughout the validation, we assume that a localization module provides position information every $1 \ \text{s}$ under the following localization uncertainty $w$:
\begin{equation}
    p_w(w) \sim U(-3, 3), ~\mathcal{W}=\{w ~|~ -3\leq w \leq 3\},
\end{equation}
where $U$ represents a uniform distribution. $3$m in longitudinal position error exceeds the accuracy of lane-level positioning \cite{williams2020qualitative}, which can be achievable using off-the-shelf systems using vision \cite{mobileye} and V2I communication \cite{cohda}. It is $95 \%$ Circle of Error Probable (CEP) for commercially available GPS \cite{vbox}.

We set the prediction horizon $N=5$, which represents $5 \text{s}$ prediction as the discretization time $T_s = 1 \text{s}$.
The observer gain $L$ in \eqref{eq: observer equation} is set to $\frac{1}{4N}$. 

\subsection{Data collection and Training} \label{sec: single road segment}
We collect the data using simulation. Given the simplicity of the system \eqref{eq:system}, data collection through simulation is feasible.
The scenario is set as follows. The vehicle's initial speed is set to zero, and the testing road is a single road segment with fixed parameters outlined in Table \ref{table: param 1}. Essentially, the vehicle needs to pass an upcoming traffic light, situated $200$ m away.
\begin{table}[h]
\centering
\caption{Parameters for the single road segment scenario}
\label{table: param 1}
\begin{tabular}{ ||l||c||} 
 \hline
 Parameter & Value\\
 \hline
 Current traffic signal &  Green \\
 Traffic light cycle  & Green: $150$s, Yellow: $5$s, Red: $25$s \\ 
 Remaining time/distance & $150$s/$200$m \\
 \hline
\end{tabular}
\end{table}

For the initialization, we assume the free-flow condition.
As explained in Sec. \ref{sec: data}, we initialize the data set by executing the cruise controller in \cite[Sec. IV. C]{ejoaiv} and collecting the state and input data as described in \eqref{eq: initial dataset}. 

For the update, we assume a constant flow speed, and the front vehicle follows the constant flow speed.
To generate such a scenario, we spawn one front vehicle. This front vehicle is located ahead of the ego vehicle with a random distance, $d_0 \sim U(5,15)$. Moreover, this front vehicle has a random initial speed, $v_0^{pv} \sim U(2,15)$, and keeps its speed until the end of the scenario, i.e., $v_k^{pv} = v_0^{pv}, ~\forall k \geq 0$.
Once the traffic flow speed $v_0^{pv}$ is sampled, the time to cross the upcoming traffic light is calculated as follows:
\begin{equation} \label{eq: time to cross, single}
    T_{1} = \biggr\lceil \frac{200}{v_0^{pv}} \biggr\rceil + 1,
\end{equation}
where $\lceil \cdot \rceil$ denotes the ceiling function.
Given the dataset, we compute the terminal set $\mathcal{X}_f$ as per Sec. \ref{sec: learning constraints} and the terminal cost $V_f(\cdot)$ as per Sec. \ref{sec: learning V function}.
Integrating $\mathcal{X}_f$ and $V_f(\cdot)$ into the proposed MPC \eqref{eq:mpc reform} and \eqref{eq: mpc simplified}, we solve MPC \eqref{eq:mpc reform} and apply the optimal input to the system \eqref{eq:system}. While running the MPC \eqref{eq:mpc reform}, we record the closed-loop state and input pairs. After completing each task iteration, we augment the data using these state-input pairs.

Fig. \ref{fig:energyconsumption_single} illustrates the learning curve of the realized total energy consumption. 
The x-axis denotes the number of data augmentation processes.
Due to the uncertainty described in system \eqref{eq:system} and the randomness of the scenario parameters, the total energy consumption is a random variable. Therefore, we calculate the sample mean of the total energy consumption by performing 100 Monte Carlo simulations of the system \eqref{eq:system} in closed-loop using the MPC \eqref{eq:mpc reform} and \eqref{eq: mpc simplified} after each data augmentation process is completed. 
The results demonstrate $31.7 \%$ improvement in total energy consumption as the number of data augmentations (or task iterations) increases, eventually leading to settled performance.
\begin{figure}[ht]
\begin{center}
\includegraphics[width=0.65\linewidth,keepaspectratio]{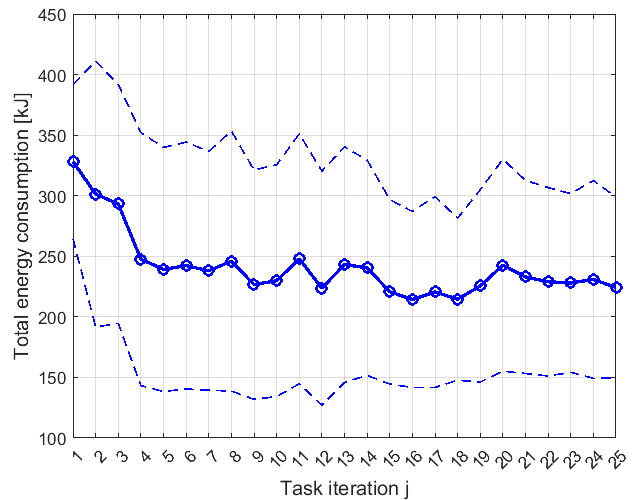}
\vspace{-1.0em}
\caption{Total Energy Consumption \textcolor{black}{Decreases} with Increasing Data Size. $100$ Monte Carlo simulations for each task iteration.}
\label{fig:energyconsumption_single}
\end{center}
\end{figure}

\subsection{Simulation results} \label{sec: sim result}
We conduct simulation studies to investigate the performance of the proposed MPC \eqref{eq:mpc reform} and \eqref{eq: mpc simplified}. We compare the proposed algorithm with two baseline algorithms described in Sec. \ref{sec: baseline}.

The simulations are conducted across four distinct scenarios, each designed to capture various traffic flow conditions.
In all scenarios, the ego vehicle's initial speed is set to zero, the front vehicle is positioned $5$m ahead, and the testing road is a single road segment with fixed parameters outlined in Table \ref{table: param 1}. 
The front vehicle's controller is the cruise control algorithm in \cite{ejoaiv} and its reference speed varies across the scenarios as follows: (1) $v_k^{pv}=2.5$m$/$s, (2) $v_k^{pv}=5.0$m$/$s, (3) $v_k^{pv}=7.5$m$/$s, (4) $v_k^{pv}=10.0$m$/$s. 
The time to cross the upcoming traffic light in each scenario is calculated accordingly based on \eqref{eq: time to cross, single}.
After we conduct each simulation scenario, we evaluate the energy consumption defined in \eqref{eq: energy definition}.

Energy consumption results are described in Table \ref{table: energy eval simul}. The unit for the energy consumption is kJ. We tune the reference speed of the cruise control algorithm and also configure the terminal time and the regularization parameter in \cite[eq. (4)]{bae2019VIL}, ensuring that the completion times for both algorithms match that of the proposed algorithm. 
\begin{table}[h]
\centering
\caption{Comparison of Energy Consumption [kJ] \\ Across Different Algorithms and front Vehicle Speeds}
\label{table: energy eval simul}
\begin{tabular}{ |l||c|c|c|c|} 
 \hline
\diagbox[width=\dimexpr \textwidth/8+2\tabcolsep\relax]{Algorithm}{$v_k^{pv}$}&$2.5$m$/$s&$5.0$m$/$s&$7.5$m$/$s&$10$m$/$s\\
 \hline
 Cruise control&$200.26$&$213.94$&$263.45$&$326.84$\\
 Algorithm \cite{bae2019VIL}&$196.11$&$163.49$&$207.50$&$271.08$\\ 
 Proposed &$170.83$&$145.76$&$171.77$&$255.48$\\
 \hline
\end{tabular}
\end{table}\\
Compared to the cruise control algorithm, the proposed algorithm improves energy \textcolor{black}{efficiency} by an average of $25.8 \%$. Compared to the algorithm in \cite{bae2019VIL}, the proposed algorithm improves energy \textcolor{black}{efficiency} by an average of $11.7 \%$. 
It is noteworthy that as the flow speed increases, the energy consumption tends to increase for all algorithms.
\textcolor{black}{As speed increases, more kinetic energy is lost during vehicle deceleration, approximately following a quadratic function of the speed. Additionally, higher speeds result in greater energy loss due to air drag and rolling resistance.}
This result is aligned with \cite[Table 1]{sciarretta2015optimal} and justifies implementing traffic light crossing time constraints in Sec. \ref{sec: tlctc}.

Fig.\ref{fig:simresult} presents the simulation results when $v_k^{pv}=7.5$m$/$s. We compare the proposed algorithm with the algorithm in \cite{bae2019VIL}. 
A comparative analysis between the proposed algorithm and the algorithm in \cite{bae2019VIL} is conducted. 
While the high-level speed planner in \cite{bae2019VIL} computes speed references for the entire task horizon at each time step, for the sake of comparison, the planned speed reference in Fig.\ref{fig:simresult}.(b) is truncated. 

\begin{figure}[ht]
\begin{center}
\includegraphics[width=0.75\linewidth,keepaspectratio]{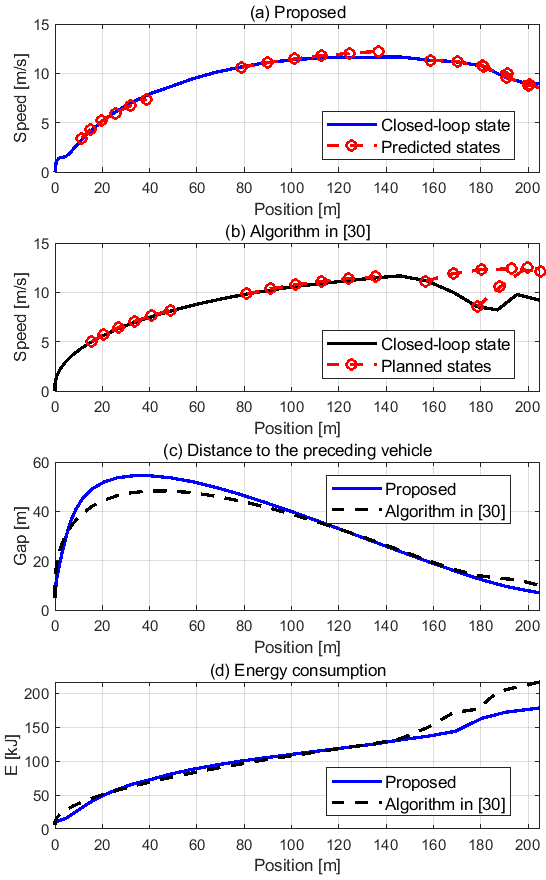}
\vspace{-1.0em}
\caption{Simulation results when $v_k^{pv}=7.5$m$/$s. (a) and (b): Vehicle speed histories: Comparison between Proposed Algorithm and Algorithm in \cite{bae2019VIL}. Closed-loop states and Predicted or Planned states at four different time steps. (c) Distance to front Vehicle: Comparison between two algorithms. (d) Energy Consumption: Comparison between two algorithms.}
\label{fig:simresult}
\end{center}
\end{figure}

As shown in Fig.\ref{fig:simresult}.(a), the predicted state of the proposed algorithm closely aligns with the closed-loop state. In particular, the last two predictions are difficult to distinguish as they overlap. 
In contrast, as shown in Fig.\ref{fig:simresult}.(b), the planned speed reference does not consistently match the closed-loop state. 
This disparity arises from the small distance to the front vehicle and the free-flow traffic assumption embedded within the speed planner in \cite{bae2019VIL}.
The distance to the front vehicle (or gap) is illustrated in Fig.\ref{fig:simresult}.(c).
When the gap is large, allowing the ego vehicle to drive in free-flow traffic conditions, the planned speed aligns with the closed-loop speed. However, when the gap becomes small, violating the free-flow assumption, the tracking controller in \cite{bae2019VIL} cannot track the planned speed reference with a small tracking error to avoid collision with the front vehicle.
This phenomenon does not happen as the collision avoidance constraints are considered when the terminal set $\mathcal{X}_f$ is computed as Algorithm \ref{alg: RNTW}.

As the objective of the tracking controller is not to minimize the energy consumption but to minimize tracking error and avoid collision with the front vehicle\cite{bae2019VIL}, the disparity between the planned and the closed-loop speed results in energy-inefficient motion.
As illustrated in Fig.\ref{fig:simresult}.(d), the energy consumption of the algorithm in \cite{bae2019VIL} becomes worse than that of the proposed algorithm when the gap decreases.

\subsection{Experiement setup and results}
We conduct a vehicle-in-the-loop experiment to investigate the performance of the proposed MPC \eqref{eq:mpc reform} and \eqref{eq: mpc simplified}. We compare the proposed algorithm with the algorithm in \cite{bae2019VIL}.
We utilized a retrofitted Hyundai Ioniq 5 as the test vehicle maneuvering around the actual testing site. We simulated four traffic lights and the front vehicle.
An overview of our test vehicle and scenario is presented in Fig. \ref{fig:test}.
\begin{figure}[ht]
\begin{center}
\includegraphics[width=0.75\linewidth,keepaspectratio]{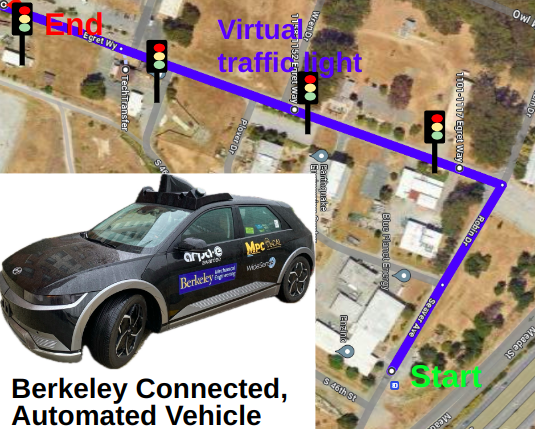}
\vspace{-1.0em}
\caption{Test vehicle and scenario}
\label{fig:test}
\end{center}
\end{figure}

For each scenario, traffic flow speed, light cycles, and locations were randomly selected. Once all relevant parameters, such as traffic flow speed, light cycles, and locations, are available, we utilize the high-level green wave search module described in Sec. \ref{sec: tlctc} to compute the time required to cross each traffic light while adhering to the prescribed traffic flow speed and rules governing traffic lights.
The parameters for each scenario are described in Table \ref{table: random scenarios}:
\begin{table}[ht]
\centering
\caption{Randomly sampled flow speed and traffic light locations, \\ and Computed time to cross the traffic lights}
\label{table: random scenarios}
\begin{tabular}{ |l|c|c|c|} 
 \hline
 {}&\makecell{Traffic flow \\ speed [m/s]}&\makecell{Traffic light \\ locations [m]} &\makecell{Time to cross\\ each traffic light[s]} \\
 \hline
 Scenario 1 &$4.312$&$[102, 245, 378, 484]$&$[24, 57, 88, 113]$\\
 Scenario 2 &$5.451$&$[126, 260, 381, 506]$&$[26, 48, 70, 94]$\\ 
 Scenario 3 &$3.433$&$[127, 227, 343, 450]$&$[37, 62, 95, 127]$\\
 \hline
\end{tabular}
\end{table}\\
In all scenarios, the front vehicle is positioned $5$m ahead and controlled using the cruise control algorithm to track the traffic flow speed. 
Moreover, the initial condition of the ego vehicle is idle speed, meaning no driver pedal inputs. Starting from idle speed, the vehicle needs to complete the given route under localization uncertainty while crossing each traffic light within the predetermined time and avoid collision with the front vehicle.


The terminal time and the regularization parameter in \cite[eq. (4)]{bae2019VIL} is tuned to show a similar travel time to the proposed algorithm.
After we conduct each test scenario, we evaluate the energy consumption defined in \eqref{eq: energy definition}.
The change of $E_\text{stor}$ is calculated using the equipped voltage and current sensors, while the change of $E_\text{kin}$ is calculated by a change in kinetic energy between the initial and the end states.

The speed profiles are presented in Fig.\ref{fig:energyconsumption_test}.
In all scenarios, the proposed algorithm crosses each traffic light without stopping, ensuring a consistent completion of the given route with minimal fluctuations, as depicted in Fig. \ref{fig:energyconsumption_test}.
In contrast, the algorithm \cite{bae2019VIL} exhibits occasional stops along the route.
Though the algorithm in \cite{bae2019VIL} ensures timely crossing of traffic lights during the planning phase, the tracking controller fails to follow the planned speed reference due to tracking error, time latency, and, most importantly, collision avoidance as described in Sec. \ref{sec: sim result}.
Consequently, when the vehicle controlled by the algorithm \cite{bae2019VIL} comes to a stop before the traffic light, it waits until the light changes to green. This delay may lead to the vehicle needing to accelerate to meet the predetermined time for crossing the next traffic light, resulting in energy-inefficient motion.

\begin{figure}[ht]
\begin{center}
\includegraphics[width=0.75\linewidth,keepaspectratio]{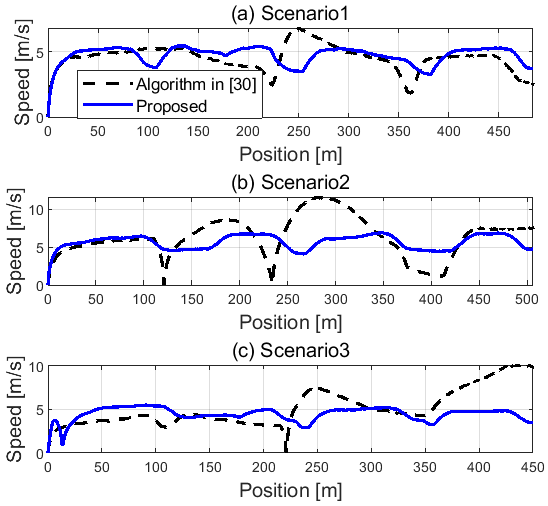}
\vspace{-1.0em}
\caption{Closed-loop speed trajectory of the ego vehicle.}
\label{fig:energyconsumption_test}
\end{center}
\end{figure}

The energy consumption results are described in Table \ref{table: energy eval}.
Compared to the algorithm in \cite{bae2019VIL}, the proposed algorithm shows an improvement of $19.4\%$ in average energy \textcolor{black}{efficiency}. 
In the best-case scenario (scenario 2), the proposed algorithm exhibits a $47.5\%$ improvement over the algorithm in \cite{bae2019VIL}, while in the worst-case scenario (scenario 1), it performs $1.5\%$ worse than the algorithm in \cite{bae2019VIL}. 
As illustrated in Fig. \ref{fig:energyconsumption_test}.(a), in scenario 1, the speed profiles of both the proposed algorithm and the algorithm in \cite{bae2019VIL} exhibit similarity, leading to comparable energy consumption results. However, as shown in Fig. \ref{fig:energyconsumption_test}.(b), in scenario 2, the vehicle controlled by the algorithm in \cite{bae2019VIL} experiences multiple stops, consequently leading to excessive energy consumption.
In contrast, the vehicle controlled by the proposed algorithm crosses the traffic light without stopping, resulting in energy-efficient motion.
Considering the predetermined time in Table. \ref{table: random scenarios}, the ego vehicle should cross the final traffic light within 113s, 94s, and 127s for scenarios 1, 2, and 3, respectively. The proposed algorithm successfully completes the route within the maximum travel time for all scenarios, while the algorithm in \cite{bae2019VIL} does not.
\begin{table}[h]
\centering
\caption{Closed-loop performance of algorithms. \\ $[ \text{Scenario 1}, \text{Scenario 2}, \text{Scenario 3} ]$ }
\label{table: energy eval}
\begin{tabular}{ |l||c|c|} 
 \hline
 Algorithm&Energy consumption[kJ]& \makecell{Travel time[s]} \\
 \hline
 Algorithm \cite{bae2019VIL} &$[192.2,328.0,214.1]$&$[114.8,134.7,127.6]$\\
 Proposed&$[195.1,171.1,188.1]$&$[106.6,92.7,109.6]$\\
 \hline
\end{tabular}
\end{table}

\section{\textcolor{black}{Limitation}}
\textcolor{black}{In this paper, we focus exclusively on longitudinal positioning uncertainty as we consider the longitudinal control of the vehicle. We model this uncertainty using an invariant probability distribution $p_w(w)$. However, various factors such as speed, road conditions, or sensor accuracy can influence longitudinal positioning uncertainty, suggesting the need for a distribution that accounts for these variables. Our current model does not consider these variations, leading to a limitation: it may not accurately represent uncertainty under different conditions. This limitation impacts the objective function, specifically the expected energy consumption outlined in Equation \eqref{eq:ftocp}. Consequently, the theory and approach used to design the controller might not fully capture the dynamics under varying conditions.}

\textcolor{black}{Additionally, the experimental setup in this study involves a single driving state for the front vehicle. Further testing is necessary under more dynamic conditions to validate the proposed method's robustness, safety, and energy efficiency. These conditions include scenarios where the front vehicle drives at variable speeds, performs cut-ins and cut-outs, and encounters sudden changes in driving intentions. Future research should include relevant experiments or a deeper discussion to address these concerns.}

\textcolor{black}{Another limitation this study acknowledges is the assumption of negligible V2I communication errors, with no packet loss or latency. This assumption is stated at the beginning of Sec. \ref{sec: problem setup} and reiterated here to highlight its impact on the controller's performance. In real-world applications, communication imperfections could affect the controller's effectiveness, necessitating further investigation and adjustments for these factors.}

\section{Conclusion and Future work}
This paper introduces a data-driven Model Predictive Control (MPC) framework tailored for energy-efficient urban road driving in the context of connected, automated vehicles while considering localization uncertainty. By addressing the challenges of longitudinal speed planning and control on roads with traffic lights and front vehicles, our proposed MPC offers a unified solution that minimizes total energy consumption while ensuring timely compliance with traffic laws and traffic light signals.

Incorporating a terminal cost function and constraints, learned from state-input data, enables effective energy optimization across the long task horizons while solving the proposed MPC with a short horizon. The proposed algorithm has been investigated via simulations and vehicle tests. The vehicle tests demonstrate the practical effectiveness of our approach, showcasing a significant $19.4\%$ improvement in average energy \textcolor{black}{efficiency} compared to conventional methods.

As an extension of the current work, we aim to enhance the robustness and adaptability of our system to accommodate the dynamic traffic environment, such as variations in traffic light. Moreover, we aim to delve into a multi-agent control for connected vehicles through V2V communication.

\section*{ACKNOWLEDGMENT}
This research work presented herein is funded
by the Advanced Research Projects Agency-Energy 
(ARPA-E), U.S. Department of Energy under DE-AR0000791. 

\bibliographystyle{IEEEtran}
\bibliography{IEEEabrv,reference}

\appendix
\subsection{Controller Implementation}
The proposed standard MPC \eqref{eq:mpc reform} is a convex optimization problem.
The stage cost is convex as $\mathbf{P} \succeq 0$ \eqref{eq: energy cost regression}. 
The system equation is linear, and state/input constraints \eqref{eq:constraints} are convex. 
$\mathcal{C}(\hat{s}^{pv}_{i|k}, \hat{v}^{pv}_{i|k})$ is a convex set from \eqref{eq: collision avoidance constraint}.
The set $\mathcal{W}$ is convex \eqref{eq: gps error}, and the sets $\mathcal{S}_{t_\text{red}}$ and $\mathcal{P}_{t_\text{green}}$ are convex polytopes as they are calculated via convex hull operation of a finite number of points as described in line 9 of Algorithm \ref{alg: RNTW}. 
Since the Pontryagin difference between two convex polytopes is convex, the fourth and the fifth constraints are convex.
The rest of the constraints are linear equalities and inequalities which are convex.
As the proposed standard MPC \eqref{eq:mpc reform} is a convex optimization problem, the simplified MPC \eqref{eq: mpc simplified} is also a convex optimization problem.
To implement the proposed MPC \eqref{eq:mpc reform} and \eqref{eq: mpc simplified}, we utilize CVXPY \cite{diamond2016cvxpy} as a modeling language and MOSEK as a solver \cite{mosek}.
For the vehicle test, we implement the proposed algorithm in ROS2 \cite{ros2}.

\subsection{Proof of Theorem \ref{thm1}} \label{app. thm}
All data points in $\mathbf{X}$ and $\mathbf{U}$ satisfy the constraints \eqref{eq:constraints} by construction. We prove the rest of the claim by induction

For $p=0$, by definition of the robust controllable set \cite[Def.10.18]{borrelli2017predictive}, $\mathcal{R}_0 =\mathcal{G} \medcap \mathcal{C}(\hat{s}_{t+N|k}^{pv},\hat{v}_{t+N|k}^{pv})$ is a $0$-step robust controllable set of the system \eqref{eq: observer equation} perturbed by the noise \eqref{eq: lumped noise} for the target convex set $\mathcal{G}$ subject to the constraints \eqref{eq:constraints} and the collision avoidance constraint \eqref{eq: collision avoidance constraint} given the predicted behavior of the front vehicle $\hat{s}_{t+N|k}^{pv},\hat{v}_{t+N|k}^{pv}$.

Suppose that for some $p \geq 0$, $\mathcal{R}_p$, declared in line 11 at iteration $p-1$, is an $p$-step robust controllable set of the system \eqref{eq: observer equation} perturbed by the noise \eqref{eq: lumped noise} for the target convex set $\mathcal{G}$ subject to the constraints \eqref{eq:constraints} and \eqref{eq: collision avoidance constraint} given the predicted behavior of the front vehicle $\{\hat{s}_{i|k}^{pv}\}_{i=t+N-p}^{t+N},\{\hat{v}_{i|k}^{pv}\}_{i=t+N-p}^{t+N}$.

At line 5, we find data points such that the state $[\mathbf{X}]_j$ of the system \eqref{eq: observer equation} can be robustly steered to $\mathcal{R}_i$. 
subsequently, at line $6$, the collision avoidance satisfaction is checked at the state $[\mathbf{X}]_j$.
Let $\{(\mathbf{\Tilde{x}}_s, ~\Tilde{u}_s)\}_{s=0}^{N_\text{s}}$ denote a set of the data points that satisfy both conditions at line 5 and 6. 
Then, by definitions \cite[Def.10.15 \& 18]{borrelli2017predictive}, $\{\mathbf{\Tilde{x}}_s\}_{s=0}^{N_\text{s}}$ are elements of the $p+1$-step robust controllable set. 

Now, we prove convex combinations of $\{\mathbf{\Tilde{x}}_s\}_{s=0}^{N_\text{s}}$ are also elements of the $p+1$-step robust controllable set.
First, we show that convex combinations of data points $\{(\mathbf{\Tilde{x}}_s, ~\Tilde{u}_s)\}_{s=0}^{N_\text{s}}$ satisfy the constraints \eqref{eq:constraints}. This is true because $\mathcal{X}$ and $\mathcal{U}$ are convex sets.
Second, we prove that convex combinations of data points $\{\mathbf{\Tilde{x}}_s \}_{s=0}^{N_\text{s}}$ satisfy the condition at line 6, i.e., the given collision avoidance constraint $\mathcal{C}(\hat{s}_{t+N-p-1|k}^{pv}, \hat{v}_{t+N-p-1|k}^{pv})$. This is also true because the given collision avoidance constraint represents a halfspace, which is convex \cite{boyd2004convex}.
Third, we prove that convex combinations of $\{\mathbf{\Tilde{x}}_s\}_{s=0}^{N_\text{s}}$ can be robustly steered to $\mathcal{R}_p$.
$\mathcal{R}_p$ is convex as it was constructed by the convex hull operation in line 11 at iteration $p-1$. The following holds for all $0 \leq \lambda \leq 1$:
\begin{equation*}
    \begin{aligned}
        & ~~~~~~~~\forall s_1, s_2 \in \{0,\cdots,N_\text{s}\},\\
        & ~~~~~~~~\mathbf{A}\mathbf{\Tilde{x}}_{s_1} + \mathbf{B}\Tilde{u}_{s_1} + \mathbf{D}n_{s_1} \in \mathcal{R}_p, \forall n_{s_1} \in 2L\mathcal{W},\\
        & ~~~~~~~~\mathbf{A}\mathbf{\Tilde{x}}_{s_2} + \mathbf{B}\Tilde{u}_{s_2} + \mathbf{D}n_{s_2} \in \mathcal{R}_p, \forall n_{s_2} \in 2L\mathcal{W}.\\
        & \implies \forall s_1, s_2 \in \{0,\cdots,N_\text{s}\}, \forall n_{s_1}, n_{s_2}  \in 2L\mathcal{W}, \\
        & ~~~~~~~~ \mathbf{\Tilde{x}}_c \coloneqq \lambda\mathbf{\Tilde{x}}_{s_1} + (1-\lambda)\mathbf{\Tilde{x}}_{s_2},\\
        & ~~~~~~~~ \mathbf{\Tilde{u}}_c \coloneqq \lambda\mathbf{\Tilde{u}}_{s_1} + (1-\lambda)\mathbf{\Tilde{u}}_{s_2}, \\
        & ~~~~~~~~ n_{c} \coloneqq \lambda n_{s_1} + (1-\lambda)n_{s_2}, \\
        & ~~~~~~~~ \mathbf{A}\mathbf{\Tilde{x}}_c + \mathbf{B}\mathbf{\Tilde{u}}_c + \mathbf{D}n_{c} \in \mathcal{R}_p ~(\because \mathcal{R}_p \text{ is convex.}).
    \end{aligned}
\end{equation*}
Moreover, as $\mathcal{W}$ is a convex set, $n_{c}$ belongs to the set $2L\mathcal{W}$ and can represent all realizations of noise in $2L\mathcal{W}$.
Thus, all convex combinations of $\{\mathbf{\Tilde{x}}_s\}_{s=0}^{N_\text{s}}$ are also elements of the $p+1$-step robust controllable set. Therefore, $\mathcal{R}_{p+1}$, constructed in line 11 by convex hull operation on the set $\mathbb{X} = \textbf{A}$, is an $i+1$-step robust controllable set when $\mathcal{R}_p$ is an $p$-step robust controllable set.
By induction, the claim is proved.

\subsection{Constraint reformulation} \label{app. constraint reform}
From the dynamics of the estimated and the nominal state in \eqref{eq:MPC} and \eqref{eq: lumped noise}, we have that:
\begin{equation*}
\begin{aligned}
    & \mathbf{e}_{i|k} = \hat{\mathbf{x}}_{i|k} - \bar{\mathbf{x}}_{i|k}, \\
    & \mathbf{e}_{i+1|k} = \mathbf{A}\mathbf{e}_{i|k} + \mathbf{D}n_i, ~n_i \sim p_n(n).
\end{aligned}
\end{equation*}
From the initial constraint in \eqref{eq:MPC}, $\mathbf{e}_{0|k} = 0$.
Moreover, $\mathbf{A}^i \mathbf{D} = \mathbf{D}$ by calculation.
Thus, for all $i \geq 1$, we have that:
\begin{equation} \label{eq: error nominal estimate}
\begin{aligned}
    & \mathbf{e}_{i|k} = \sum_{k=0}^{i-1} \mathbf{A}^k \mathbf{D}n_k = \mathbf{D} \sum_{k=0}^{i-1}n_k.
\end{aligned}
\end{equation}
From \eqref{eq: lumped noise}, we know that $n_k \in 2L\mathcal{W}, ~\forall k \geq 0$. Thus, for all $i \geq 1$, we have that:
\begin{equation} \label{eq: error in 2lidw}
\begin{aligned}
    & \mathbf{e}_{i|k} = \mathbf{D} \sum_{k=0}^{i-1}n_k \in 2Li\mathbf{D}\mathcal{W}.
\end{aligned}
\end{equation}
Based on the result in \eqref{eq: error in 2lidw}, the constraints in \eqref{eq:MPC} is reformulated with respect to the nominal state.

First, we consider the state and input constraints in \eqref{eq:MPC}. These constraints do not depend on the noisy position by definition in \eqref{eq:constraints}. Moreover, the speed of the nominal state propagated from $\bar{\mathbf{x}}_{0|k}=\hat{\mathbf{x}}_k$ is identical to that of the estimated state. Thus, we have that:
\begin{equation*}
\begin{aligned}
    & \bar{\mathbf{x}}_{i|k} \in \mathcal{X}, ~ u_{i|k} \in \mathcal{U} \iff \hat{\mathbf{x}}_{i|k} \in \mathcal{X}, ~ u_k \in \mathcal{U}.
\end{aligned}
\end{equation*}
Second, we reformulate the collision avoidance constraints, the terminal constraints, and the traffic light constraints in \eqref{eq:MPC}. From \eqref{eq: error in 2lidw}, the following holds for any given set $\mathcal{S}$:
\begin{equation}
    \bar{\mathbf{x}}_{i|k} \in \mathcal{S}\ominus 2Li\mathbf{D} \mathcal{W} \implies \hat{\mathbf{x}}_{i|k} \in \mathcal{S}.
\end{equation}
Thus, the following hold for all $n_k \in 2L\mathcal{W}$:
\begin{equation} \label{eq: constraint reform result}
\begin{aligned}
    & \bar{\mathbf{x}}_{i|k} \in \mathcal{C}(\hat{s}_{i|k}^{pv},\hat{v}_{i|k}^{pv})\ominus 2Li\mathbf{D}\mathcal{W} \implies \hat{\mathbf{x}}_{i|k} \in \mathcal{C}(\hat{s}_{i|k}^{pv},\hat{v}_{i|k}^{pv}), \\
    & \bar{\mathbf{x}}_{N|k} \in \mathcal{X}_f\ominus 2LN\mathbf{D}\mathcal{W} \implies \hat{\mathbf{x}}_{N|k} \in \mathcal{X}_f, \\
    & \bar{\mathbf{x}}_{i|k} \in \mathcal{L}(s_\mathrm{tl})\ominus(2Li+1)\mathbf{D}\mathcal{W} \implies \hat{\mathbf{x}}_{i|k} \in \mathcal{L}(s_\mathrm{tl})\ominus\mathbf{D}\mathcal{W}.
\end{aligned}
\end{equation}

\subsection{Integration of \eqref{eq: initial value function} to \eqref{eq:mpc reform}} \label{app: integration MPC}
The minimization problem \eqref{eq: initial value function} can be incorporated into the problem \eqref{eq:mpc reform} and form a single unified optimization problem as follows:
\begin{equation} \label{eq: mpc unified}
\begin{aligned}
    & \Tilde{J}_\text{MPC} (\hat{\mathbf{x}}_{k}, \{\hat{s}^{pv}_{i|k}\}_{i=0}^{N}, \{\hat{v}^{pv}_{i|k}\}_{i=0}^{N}) = \\
    & \min_{\substack{u_{0:N-1|k},\\ \mathbf{s}_f, \\\bm{\lambda}_{1:M},\\ \bm{\lambda}^\mathcal{O}_{1:M}}} \sum_{i=0}^{N-1} \ell(\bar{\mathbf{x}}_{i|k}, ~ u_{i|k}) +\sum_{m=1}^M p_m V_m +  M_f \mathbf{s}_f\\
    & ~~~~ \textnormal{s.t.,} \, ~~ \bar{\mathbf{x}}_{i+1|k} =\mathbf{A} \bar{\mathbf{x}}_{i|k} + \mathbf{B} u_{i|k} \\
    & \qquad ~~~~~  \bar{\mathbf{x}}_{0|k} = \hat{\mathbf{x}}_{k}, \\
    & \qquad ~~~~~  \bar{\mathbf{x}}_{i|k} \in \mathcal{X}, ~ u_{i|k} \in \mathcal{U},\\
    & \qquad ~~~~~  \bar{\mathbf{x}}_{i|k} \in \mathcal{C}(\hat{s}^{pv}_{i|k}, \hat{v}^{pv}_{i|k}) \ominus 2Li\mathbf{D}\mathcal{W} \\
    & \qquad ~~~~~  \bar{\mathbf{x}}_{N|k} \in (\mathcal{S}_{t_\text{red}} \medcap \mathcal{P}_{t_\text{green}}) \ominus 2LN\mathbf{D}\mathcal{W}\oplus \mathbf{s}_f,\\
    & \qquad ~~~~~  \bar{\mathbf{x}}_{i|k} \in \mathcal{L}(s_\mathrm{tl}) \ominus (2Li+1)\mathbf{D}\mathcal{W}, \text{if} \, c_{l^\star,k+i}=\text{red}, \\
    & \qquad ~~~~~  \mathbf{s}_f \geq 0 \\
    &\qquad ~~~~~  V_m = \mathbf{J} \bm{\lambda}_m + \mathbf{J}_\mathcal{O} \bm{\lambda}^\mathcal{O}_m,\\
    &\qquad ~~~~~  \mathbf{X} \bm{\lambda}_m + \mathbf{V}_\mathcal{O} \bm{\lambda}^\mathcal{O}_m = \bar{\mathbf{x}}_{N|k} + \mathbf{v}_N^m - \mathbf{s}_\mathrm{tl}\\
    & \qquad ~~~~~  \bm{\lambda}_m \geq \bm{0}, ~ \bm{1}^\top \bm{\lambda}_m = 1, \\
    & \qquad ~~~~~  \bm{\lambda}^\mathcal{O}_m \geq \bm{0}, \,\, \bm{1}^\top \bm{\lambda}^\mathcal{O}_m = 1 \\
    & \qquad ~~~~~  i \in \{0,\cdots,N-1\}, ~  m \in \{1,\cdots,M\}.
\end{aligned}
\end{equation}

\subsection{Proof of Theorem \ref{thm3}} \label{app. thm3}
If $k+N \leq T_{l^\star}$, the MPC controller in \eqref{eq:mpc reform} is employed. Let $u_{0:N-1|k}^\star$ denote the optimal input sequence of \eqref{eq:mpc reform}. From the terminal constraint and \eqref{eq: constraint reform result}, the system \eqref{eq: observer equation} controlled by $u_{0:N-1|k}^\star$ robustly steers to the terminal set $\mathcal{X}_f = \mathcal{S}_{t_\text{red}} \medcap \mathcal{P}_{t_\text{green}}$ from $\hat{\mathbf{x}}_k$, namely:
\begin{equation}
    \hat{\mathbf{x}}_{N|k} \in \mathcal{X}_f = \mathcal{S}_{t_\text{red}} \medcap \mathcal{P}_{t_\text{green}}, ~\forall n_k \in 2L\mathcal{W}.
\end{equation}
As $\mathcal{P}_{t_\text{green}}$ is a $t_\text{green}$-step robust controllable set to the target set $\hat{\mathcal{G}}_p$, there exists a control input sequence that robustly steers to $\hat{\mathcal{G}}_p$ from $\hat{\mathbf{x}}_{N|k}$, i.e., $\hat{\mathbf{x}}_{N+t_\text{green}|k} \in \hat{\mathcal{G}}_p$. Moreover, as $\hat{\mathcal{G}}_p={\mathcal{G}}_p\ominus\mathbf{D}\mathcal{W}$ described in \textbf{Remark 8}, $\hat{\mathbf{x}}_{N+t_\text{green}|k} \in \hat{\mathcal{G}}_p \implies {\mathbf{x}}_{N+t_\text{green}|k} \in {\mathcal{G}}_p$, which represents that the system \eqref{eq:system} robustly steers to ${\mathcal{G}}_p$, the region after the upcoming traffic light. This proves the claim for $k+N \leq T_{l^\star}$.

If $k+N > T_{l^\star}$, the MPC controller in \eqref{eq: mpc simplified} is employed. Let $u_{0:N_s-1|k}^\star$ denote the optimal input sequence of \eqref{eq: mpc simplified}. From the terminal constraint and \eqref{eq: constraint reform result}, the system \eqref{eq: observer equation} controlled by $u_{0:N_s-1|k}^\star$ robustly steers to the terminal set $\hat{\mathcal{G}}_p$. Similar to the previous case, this implies that the system \eqref{eq:system} robustly steers to ${\mathcal{G}}_p$, the region after the upcoming traffic light. This proves the claim for $k+N > T_{l^\star}$.

The claim is proved in both cases.

\end{document}